\documentclass[accepted]{uai2025} % after acceptance, for a revised version; 

\usepackage[american]{babel}
\usepackage{natbib} % has a nice set of citation styles and commands
\usepackage{mathtools} % amsmath with fixes and additions
\usepackage{booktabs} % commands to create good-looking tables
\usepackage{tikz} % nice language for creating drawings and diagrams

\usepackage{amsthm} %for Proof enviroment
\usepackage{amsfonts} %For writting mathbb
\usepackage{algorithm} %For pseudo code
\usepackage{algpseudocode} %For pseudo code
\usepackage{pgfplots} %For tikz made in CausalDisco
\pgfplotsset{compat=1.18} 
\usetikzlibrary{arrows.meta,arrows,shapes,decorations,automata,backgrounds,petri} %For tikz made in CausalDisco
\usetikzlibrary{decorations.markings}
\usepackage{multirow}
\usepackage{changepage}   % for the adjustwidth environment

\newcommand{\N}{\mathbb{N}}
\newcommand{\G}{\mathcal{G}}%For graphs
\newcommand{\E}{\mathcal{E}} %For equivalence classes of graphs
\newcommand{\D}{\mathbf{D}} %For data notation
\newcommand{\K}{\mathcal{K}_\tau} %For background knowledge
\newcommand{\EK}{\mathcal{E}^{\mathcal{K}_\tau}}
\newcommand{\TS}{S_\tau} %TBIC
\newcommand{\Ts}{s_\tau} %local TBIC
\newcommand{\V}{\mathbf{V}}
\newcommand{\Ed}{\mathbf{E}}
\newcommand{\EKS}{\EK_*}
\newcommand{\HH}{\mathcal{H}}%For graphs
\DeclareMathOperator*{\argmax}{argmax} % thin space, limits underneath in displays

\newtheorem{definition}{Definition}%[section]
\newtheorem{lemma}{Lemma}%[section]
\newtheorem{theorem}{Theorem}%[section]
\newtheorem{proposition}{Proposition}%[section]

\newtheorem{innercustomlem}{Lemma}
\newenvironment{customlem}[1]
  {\renewcommand\theinnercustomlem{#1}\innercustomlem}
  {\endinnercustomlem}

\newtheorem{innercustomthm}{Theorem}
\newenvironment{customtheorem}[1]
  {\renewcommand\theinnercustomthm{#1}\innercustomthm}
  {\endinnercustomthm}

\newcommand{\GS}{\mathcal{G}^*}
\newcommand{\indG}{\perp\!\!\!\perp_\G}
\newcommand{\indGS}{\perp\!\!\!\perp_{\GS}}
\newcommand{\nindG}{\not\!\perp\!\!\!\perp_\G}
\newcommand{\nindGS}{\not\!\perp\!\!\!\perp_{\GS}}
\newcommand{\paG}{\text{Pa}^{\G}}
\newcommand{\paGS}{\text{Pa}^{\GS}}
\newcommand{\deG}{\text{De}^{\G}}
\newcommand{\deGS}{\text{De}^{\GS}}
\newcommand{\ndeG}{\text{NDe}^{\G}}

\DeclareMathAlphabet\mathbfcal{OMS}{cmsy}{b}{n} %bold caligraphical letters

\title{Score-Based Causal Discovery with Temporal Background Information}

\author[1]{\href{mailto:<tobias.ellegaard@sund.ku.dk>?Subject=Paper: Score-Based Causal Discovery with Temporal Background Information}{Tobias~E.~Larsen}{}}
\author[1]{Claus~T.~Ekstrøm}
\author[1]{Anne~H.~Petersen}
\affil[1]{%
    Section of Biostatistics, Department of Public Health\\
    University of Copenhagen\\
    Copenhagen, Denmark
}
  
\begin{document}
\maketitle
\pagestyle{numbered}
\begin{abstract}
Temporal background information can improve causal discovery algorithms by orienting edges and identifying relevant adjustment sets. We develop the Temporal Greedy Equivalence Search (TGES) algorithm and terminology essential for score-based causal discovery with tiered background knowledge. TGES learns a restricted Markov equivalence class of directed acyclic graphs (DAGs) using observational data and tiered background knowledge. To construct TGES we formulate a scoring criterion that accounts for tiered background knowledge. We establish theoretical results for TGES, stating that the algorithm always returns a tiered maximally oriented partially directed acyclic graph (tiered MPDAG) and that this tiered MPDAG contains the true DAG in the large sample limit. We present a simulation study indicating a gain from using tiered background knowledge and an improved precision-recall trade-off compared to the temporal PC algorithm. 
We provide a real-world example on life-course health data.

\end{abstract}

\section{Introduction}\label{sec:intro}
When aiming to estimate the causal effect of an exposure using observational data, studies will make use of knowledge of the causal structure of the data. This structure will be represented as a \textit{Directed Acyclic Graph} (DAG) \citep{hernan2020causal}, where the variables are illustrated as nodes and direct causal effects are depicted as edges.
In practice, studies often rely on the knowledge of field experts in order to 
construct a DAG.
However, former work shows that there will not always be consensus of causal relationships among experts, hence researchers might be interested in using causal discovery as a data-driven alternative or complement \citep{TheoryVsDataPetersen}. 
\\
Causal discovery seeks to estimate a causal network using empirical data. There exist two overall branches of causal discovery: constraint-based algorithms and score-based algorithms. A canonical algorithm in the constraint-based branch is the \textit{Peter Clark} (PC) algorithm \citep{PC_article}, which uses conditional independence testing to rule out causal links. For the score-based algorithms, a seminal algorithm would be the \textit{Greedy Equivalence Search} (GES) algorithm \citep{Chickering2002}, which scores graphs on their ability to fit the observed data.
\\
\\
It has been proposed to utilize the temporal information embedded in longitudinal data to improve the estimation of the causal network, as we can rule out any causal effect going backward in time \citep{Petersen2021datadriven}. Several implementations are available for a temporal extension of the PC algorithm (TPC) \citep{CausalDisco,TPCpackage}. However, the PC algorithm is known to disproportionally favor high precision of the graph \citep{petersen_ramsey_ekstrøm_spirtes_2023}, which could be an issue for subsequent causal inference, since a sparser graph imposes more constraints on the distribution and hence is more liable to introduce bias when the graph is incorrectly identified \citep{malinsky2024cautiousapproachconstraintbasedcausal}. 
\\
GES is known to have a more balanced precision-recall trade-off \citep{ComparisonOfPackages}, which makes it an attractive alternative to PC.
\\
We have developed a temporal extension of the GES algorithm, which we call \textit{Temporal Greedy Equivalence Search} (TGES), that inherits the desirable theoretical properties of GES, and proves to be a viable alternative to the temporal extension of PC.
\\
We present TGES and show that it is sound and complete in the limit of large sample sizes. Furthermore, we investigate the performance of TGES in finite samples.
\\
\\
The paper is structured as follows. After introducing preliminary graph notation and assumptions in Section \ref{sec:Graph}, we present \textit{Greedy Equivalence Search} (GES) in Section \ref{sec:GES}. In Section \ref{sec:tier_bg_kn} we define tiered background knowledge. 
In Section \ref{sec:TGES} we define a naive extension to GES which incorporates tiered background knowledge called \textit{Simple Temporal Greedy Equivalence Search} (STGES). Furthermore, we introduce the main contribution, the \textit{Temporal Greedy Equivalence Search} (TGES) algorithm, and show that it is sound and complete. In Section \ref{sec:sim_study} we conduct a simulation study which compares TGES to the performance of GES, STGES, and TPC. TGES is here shown to have considerably better recall than TPC while maintaining a precision that is on par. We showcase TGES using a real-world data example in Section \ref{sec:data_application}, and conclude with a discussion in Section \ref{sec:discussion}.

\section{Background and Notation} \label{sec:Graph}
Throughout the article we assume familiarity with the concepts of \textit{Directed Acyclic Graph} (DAG), \textit{Partially Directed Acyclic Graph} (PDAG), \textit{Complete Partially Directed Acyclic Graph} (CPDAG), \textit{d-separation}, \textit{faithfulness}, \textit{Markov equivalence class} and the\textit{ Markov property} (also known as the causal Markov assumption) \citep{CausalityPeters},. We denote a graph $\G$ as a set of \textit{nodes} $\V$ and \textit{edges} $\Ed$, such that $\G = (\V,\Ed)$. A Markov equivalence class is denoted as $\E$.
\\
\\
In this article, we assume there to be a data set consisting of $d$ random variables $\mathbf{X} = (X_1,..., X_d)$. The data set $\D$ is assumed to be an observed sample of $n$ i.i.d. observations of $\mathbf{X}$.
We assume what is commonly referred to as \textit{causal sufficiency}, namely that there are no unobserved confounders. We also assume no unobserved selection.
\\
We assume that the variables are realizations of a multivariate Gaussian distribution according to a DAG and parameters $\theta$. 
This DAG will be referred to as the true data-generating DAG or the underlying data-generating DAG and the true distribution is assumed Markov and faithful with respect to this DAG. The distribution of the underlying data-generating DAG $\G$ will be denoted as $p$ and the density as $p(\D|\theta,\G)$.
\\
Parents, children, ancestors, descendants, and non-descendants of a node $X$ in $\G$ we denote as $\text{Pa}_X^{\G},\text{Ch}_X^{\G},\text{An}_X^{\G},\text{De}_X^{\G}$, $\text{NDe}_X^{\G}$ respectively, where $X\in \text{De}_X^{\G}$ and $X\notin \text{An}_X^{\G}$. 
\\
Whenever we compare two graphs, it is assumed that they are defined on the same set of nodes.
\\
\\
We will consider PDAGs obtained by restricting a CPDAG according to background information, then \citep{meek1995causalbg} has provided 4 orientation rules which together ensure that no more edges can be oriented in a PDAG. 
These rules are referred to as Meek's rules and are given in Appendix \ref{supmat:subsec:Meeks}. The graph that results from applying the rules is a \textit{maximally oriented partially directed acyclic graph} (MPDAG) \citep{perković2017}. The rules orient edges to prevent cycles and to maintain conditional independence statements.

\section{Greedy Equivalence Search} \label{sec:GES}
A score-based algorithm estimates a graph by comparing different graph structures' ability to fit the data. The idea is that graphs embedding wrong conditional independencies will fit the data poorly. 
\\
\\
Provided a scoring criterion $S(\G,\D)$, the simplest score-based algorithm is the \textit{Best Scoring Graph} \citep{CausalityPeters}. Here, we score every graph $\G$ over the nodes $\mathbf{X}$ and choose the graph that yields the maximum score as our estimate,
\begin{equation*}
\hat{\G} = \argmax_{\G \text{ over } \mathbf{X}} S(\G,\D)    
\end{equation*}
The number of DAGs to score for the \textit{Best Scoring Graph} procedure, grows super-exponentially with the number of nodes, and the algorithm thus quickly becomes computationally infeasible \citep{OEIS_A003024}.
Instead of scoring every possible DAG to obtain an estimate, we can use the \textit{Greedy Equivalence Search} (GES) \citep{Chickering2002}.
GES starts with an empty graph and greedily adds edges to improve the score of the graph. 
The adding of edges is the first of two phases in GES and is referred to as the \textit{forward} phase. After the forward phase, GES optimizes the score by greedily removing edges, which is referred to as the \textit{backward} phase. 
\\
An equivalence class $\E'$ is contained in the forward (backward) \textit{neighboring} equivalence class of $\E$ 
if and only if there exists a $\G\in \E$ and a $\G' \in \E'$ such that $\G$ becomes $\G'$ by one single directed edge addition (removal). 
\\
Each step in the phases of GES substitutes for a neighboring equivalence class if it has a higher score. If multiple neighboring equivalence classes have a higher score, then the one with the highest score is chosen. 
\\
\\
There are several extensions to GES, such as a third \textit{turning} phase which seeks to reverse edges to improve the estimations on finite samples \citep{hauser2012characterization}, or a \textit{Fast GES} which is more feasible to use on high dimensional graphs \citep{Ramsey2017}. 

\subsection{Bayesian Information Criterion}
The \textit{Bayesian Information Criterion} (BIC) \citep{BICSchwarz} is often used as the scoring criterion, and we denote it $S_B$ and define it as
\begin{align*}
    S_B(\G,\D) = \log p(\D|\hat{\theta},\G) - \frac{\#\text{parameters}}{2} \log n
\end{align*}
where $\hat{\theta}$ is the maximum likelihood estimate. In this formulation, larger BIC means better model fit. The BIC holds several desirable properties such as being \textit{decomposable}, which means we can rewrite the score as a sum of local scores of each node given the parents of the node \citep{Chickering2002}. When computing the difference in scores between two graphs, decomposability allows us to only compute the difference in the local scores of the nodes that differ in parent sets. Furthermore, the BIC is a \textit{consistent} scoring criterion for Gaussian data \citep{Haughton}. Consistent scoring criteria disfavors graphs with independence statements not in the true distribution and graphs with too many edges. The BIC is a \textit{locally consistent} scoring criterion as well, which means the score increases as a result of adding an edge that eliminates an independence constraint not in $p$ and decreases as a result of adding an edge that does not eliminate such a constraint.
\\
BIC scores every DAG in the same Markov equivalence class the same, this property is referred to as \textit{score equivalence} and gives sense to the notation $S_B(\E,\D)$, even though we in practice will score a single $\G\in\E$ \citep{Chickering2002}. 
\\
\\
GES requires a score equivalent, decomposable, consistent, and locally consistent score criterion (Formal definitions of these are provided in Appendix \ref{supmat:subsec:BIC}), such as the BIC. These properties enable an efficient search space and that GES in the limit of large sample size correctly estimates the Markov equivalence class of the true data-generating DAG \citep{Chickering2002}. 
The results from \citet{Chickering2002} (Lemmas 9,10 and Theorems 15, 17) are reproduced below: 
\begin{theorem}(Sound and Completeness of GES)\\ 
    GES using a score equivalent, decomposable, consistent and locally consistent score criterion, results in a CPDAG, and for $n \rightarrow \infty$ the CPDAG is almost surely a sound and complete estimate of the Markov equivalence class of the true data-generating DAG. \label{thm:GES_S_C}
\end{theorem}

\section{Tiered Background Knowledge} \label{sec:tier_bg_kn}
In this section, we give a formal introduction to what we have earlier called temporal information and describe how it can be used in causal discovery. We call it \textit{tiered background knowledge}, as it is not restricted to information induced from time, although this is the main use case.
\\
\\
We start out by defining \textit{background knowledge} $\mathcal{K}= (\mathcal{R},\mathcal{F})$ as a set of required edges $\mathcal{R}$ and forbidden edges $\mathcal{F}$.
We say that a DAG $\G=(\V,\Ed)$ \textit{encodes} the background knowledge $\mathcal{K}$ if $\mathcal{R} \subseteq \Ed$ and $\mathcal{F}\cap\Ed=\emptyset$. If a graph $\G$ does not encode $\mathcal{K}$, we say that $\G$ contradicts $\mathcal{K}$. We say a CPDAG/equivalence class $\E$ is \textit{in agreement} with background knowledge $\mathcal{K}$ if there exists a $\G\in\E$ that encodes $\mathcal{K}$. Analogously, a PDAG is in agreement with $\mathcal{K}$ if the class of DAGs the PDAG represents contains a DAG that encodes $\mathcal{K}$. We say that a distribution encodes $\mathcal{K}$ when the underlying data-generating DAG of the distribution encodes $\mathcal{K}$. We assume throughout that the background knowledge is encoded by the true distribution $p$.
\\
Before defining \textit{tiered background knowledge} we give the definition of a \textit{tiered ordering}.
\begin{definition}
    (Tiered ordering, \citet{bang2023wiser})
    \\
    Let $\G$ be a PDAG with node set $\mathbf{V}$ of size $d$, and let $T \in \N, T \leq d$. A tiered ordering of the nodes in $\mathbf{V}$ is a map $\tau : \mathbf{V} \mapsto \{ 1,...,T\}^d$ that assigns each node $V \in \mathbf{V}$ to a unique tier $t \in \{1, ... ,T\}$.
\end{definition}
\begin{definition}(Tiered background knowledge, \citet{bang2023wiser})
\\
Given a tiered ordering $\tau$ we define tiered background knowledge as $\K = (\mathcal{R},\mathcal{F})$ where $\mathcal{R} = \emptyset$ and $\mathcal{F}$ is the set of forbidden directed edges $ \{A \rightarrow B\} $ such that  $ \tau(A) > \tau(B) $.
\end{definition}
Graphs that encode the tiered background knowledge are made up of two types of edges: An edge $ \{A \rightarrow B\}$ is a \textit{cross-tier} edge if $ \tau(A) < \tau(B)$ and an \textit{in-tier} edge if $ \tau(A) = \tau(B)$.

\begin{algorithm}[H]
\caption{Restrict PDAG according to $\K$ \citep{perković2017}}
\label{alg:restrict}
\begin{algorithmic}[1]
\Statex \textbf{Input:} $\K= (\mathcal{R},\mathcal{F})$ and PDAG $\mathcal{W} = (\V,\Ed)$ in agreement with $\K$
\State $\Ed' \gets \Ed$
\For{ $\left\{A \text{ --- } B\right\} \in \Ed$}
\If{$\left\{A \rightarrow B\right\} \in \mathcal{F}$}
\State Replace $\left\{A \text{ --- } B\right\}$ with $\left\{A \rightarrow B\right\}$ in $\Ed'$
\EndIf
\EndFor
\State \textbf{Output:} PDAG $\mathcal{W'}=(\V,\Ed')$
\end{algorithmic}
\end{algorithm}

If a CPDAG is in agreement with $\K$, we can \textit{restrict} it according to $\K$ (Algorithm \ref{alg:restrict}) and obtain the resulting PDAG. This enables the orientation of additional undirected edges by Meek's rules \citep{meek1995causalbg}.
By \citet{bang2023wiser}, the graph resulting from applying only Meek's rule 1 until no further change is a \textit{tiered MPDAG}, which no longer represents a Markov equivalence class, but instead a \textit{restricted equivalence class} which we will denote $\EK$.
\\
Tiered background knowledge does not affect the existence or absence of adjacencies. It only affects directions, and only where there exists an adjacency between nodes in two different tiers (cross-tier). This may indirectly affect in-tier edge orientations as well (via Meek's rule 1).
%\\
%\\

\section{The Temporal Greedy Equivalence Search Algorithm} \label{sec:TGES}
In this section we first present a naive temporal modification of GES, then introduce a new score criterion that incorporates tiered background knowledge, and following this propose a more advanced temporal extension of GES. We prove that both are sound and complete.
\subsection{Simple Temporal Greedy Equivalence Search}
The algorithm \textit{Simple Temporal Greedy Equivalence Search} (STGES) is a naive modification of GES utilizing tiered background knowledge $\K$.
First STGES runs GES to completion without taking $\K$ into account to obtain a CPDAG. This CPDAG may not, due to statistical error, be in agreement with $\K$. STGES then removes all directed edges contradicting $\K$ to obtain a PDAG that is in agreement with $\K$. Afterward, STGES restricts undirected edges according to $\K$ using Algorithm \ref{alg:restrict} and then uses Meek's rules 1-4 \citep{meek1995causalbg} to orient additional edges. This is a post-hoc background knowledge approach following ideas presented by \citet{perković2017}. 
The STGES algorithm is illustrated and pseudo code is provided in Appendix \ref{supmat:sec:alg}.
\\
In the large sample limit, STGES is sound and complete:
\begin{theorem}(Sound and Completeness of STGES)\\ 
    For $n \rightarrow \infty$ STGES almost surely results in a CPDAG which is a sound and complete estimate of the Markov equivalence class of the true data-generating DAG. \label{thm:STGES_S_C} 
\end{theorem}
The proof of Theorem \ref{thm:STGES_S_C} is provided in Appendix \ref{supmat:subsec:STGES}.
\\
\\
When using tiered background knowledge, we will often have very high trust in the validity of the information, for example, if this information is temporal. A post-hoc method such as STGES, will not reflect this, as the first part of the algorithm will completely disregard $\K$.
Therefore, we propose an alternative algorithm that utilizes $\K$ earlier on, we call this \textit{Temporal Greedy Equivalence Search} (TGES).

\subsection{Temporal Greedy Equivalence Search}
We require the score criterion in TGES to have certain properties. Like for GES, we require decomposability but will extend consistency, local consistency, and score equivalence to \textit{$\K$-consistency, local $\K$-consistency} and \textit{$\K$-score equivalence} respectively. The definition of $\K$-consistency is given in Appendix \ref{supmat:subsec:TBIC}, since combining decomposability and $\K$-consistency leads to the more interpretable property of \textit{local $\K$-consistency}. We denote independence and dependence in the distribution $p$ as $\perp\!\!\!\perp_p$ and $\not\!\perp\!\!\!\perp_p$, respectively.
\begin{definition} (Local $\K$-Consistency)\label{def:LocConsTier}
\\
Let $\G$ be any DAG that encodes $\K$, and let $\G'$ be the DAG that results from adding the edge $\left\{X \rightarrow Y\right\}$ to $\G$. A scoring criterion $S$ is $\K$-locally consistent if the following holds:
\begin{enumerate}
    \item[i.] If $\tau(X) \leq \tau(Y)$ and
    $Y \not\!\perp\!\!\!\perp_p X \mid \text{Pa}_{Y}^{\G}$, then as $n \rightarrow \infty$ $S(\G', \D, \K) > S(\G, \D, \K)$ almost surely.
    \item[ii.] If $\tau(X) \leq \tau(Y)$ and $Y \perp\!\!\!\perp_p X \mid \text{Pa}_{Y}^{\G}$, then as $n \rightarrow \infty$ $S(\G', \D, \K) < S(\G, \D, \K)$ almost surely.
    \item[iii.] If $\tau(X) > \tau(Y)$, then $S(\G', \D, \K) < S(\G, \D, \K)$.
\end{enumerate}  
\end{definition}
\begin{definition} ($\K$-Score Equivalence)\label{def:K_ScoreEq}
    \\ 
    A scoring criterion $S$ is $\K$-score equivalent if for any pair of DAGs $\G,\G'$ in a given restricted equivalence class $\EK$, we have that  
    \begin{equation*}
        S(\G,\D,\K)=S(\G',\D,\K)
    \end{equation*}
\end{definition}
We construct a scoring criterion that has exactly these properties for Gaussian data, namely the \textit{Temporal Bayesian Information Criterion} (TBIC):
\begin{definition} (Temporal Bayesian Information Criterion)\\ \label{def:TBIC}
\begin{align*}
    \TS(\G,\D,\K) = 
\begin{cases} 
S_B(\G,\D) & \text{if } \G \text{ encodes } \K \\
-\infty & \text{if } \G \text{ contradicts } \K 
\end{cases}
\end{align*}    
\end{definition}
\begin{proposition} \label{thm:TBICallProps}
    TBIC on Gaussian data is $\K$-consistent, $\K$-score equivalent, decomposable and hence also locally $\K$-consistent. 
\end{proposition}
The proof for Proposition \ref{thm:TBICallProps} is provided in Appendix \ref{supmat:subsec:TBIC}. 
\\
\\
Before presenting TGES, we will introduce neighboring restricted equivalence classes of $\EK$. An equivalence class $\Tilde{\E}^{\K}$ is contained in the forward, backward or turning neighboring equivalence classes of $\EK$ if and only if there exists a $\G\in \EK$ and a $\Tilde{\G} \in \Tilde{\E}^{\K}$ such that $\G$ becomes $\Tilde{\G}$ by one single directed edge addition, removal or reversal, respectively. We will analogously use the term \textit{neighboring graphs} for the corresponding tiered MPDAGs.
\\
\\
TGES requires a scoring criterion $S$ that is $\K$-consistent, $\K$-score equivalent, decomposable, and hence will also be locally $\K$-consistent.
\\
TGES consists of 3 phases analogously to GES: forward, backward and turning.
Like GES, TGES starts with an empty graph and then iterates through forward, backward and turning phases until no more operations can be made by either phase. Each of the phases consists of \textit{steps} which are repeated until no further changes that increase the score can be made to the graph (See Algorithm \ref{alg:TGES}). Each step in a phase consist of 3 \textit{stages} (illustrated in Figure \ref{fig:TGES}):
\\
Stage (i) evaluates the neighboring CPDAGs of the tiered MPDAG. If no neighboring CPDAG has a greater score, we conclude the phase with the current tiered MPDAG. If a neighboring CPDAG has a greater score, we choose the graph with the largest score and move to the next stage.
\\
Stage (ii) restricts the CPDAG according to $\K$ by orienting all undirected cross-tier edges according to the tiered ordering.
\\
Stage (iii) uses Meek's rule 1 to infer extra edge orientations in the graph, resulting in a tiered MPDAG. 
\\
We then go to Stage (i) again and commence a new step.
\begin{algorithm}[H]
\caption{Temporal Greedy Equivalence Search (TGES)}
\label{alg:TGES}
\begin{algorithmic}[1]
\Statex \textbf{Input:} $\D,\K$
\State $\mathcal{W} \gets $ empty graph.
\While{there exist a neighboring graph $\mathcal{W}'$ s.t. $S(\mathcal{W},\D,\K)<S(\mathcal{W}',\D,\K)$}
\While{$\exists$ edge additions s.t. score improves}
\State $\mathcal{W} \gets $Temporal\_Forward\_Step$(\mathcal{W},\D,\K)$
\EndWhile
\While{$\exists$ edge removals s.t. score improves
}
\State $\mathcal{W} \gets $Temporal\_Backward\_Step$(\mathcal{W},\D,\K$
\EndWhile
\While{$\exists$ edge reversals s.t. score improves
}
\State $\mathcal{W} \gets $Temporal\_Turning\_Step$(\mathcal{W},\D,\K)$
\EndWhile
\EndWhile
\State \textbf{Output:} $\mathcal{W}$
\end{algorithmic}
\end{algorithm}
The algorithms for Temporal\_Forward-, -Backward- and -Turning\_Step are provided in Appendix \ref{supmat:sec:alg}. 
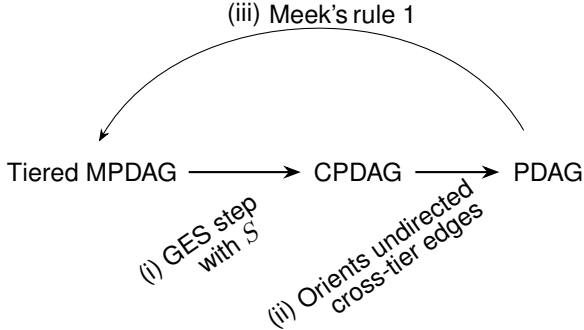
\begin{figure}[!htb]
\centering
\begin{tikzpicture}[>=Stealth, node distance=2cm, every node/.style={circle, draw, thick}]

  \tikzset{reg/.style={draw = none,font=\sffamily}}

%Nodes
  \node[reg] (TMPDAG) at (-2,0) {Tiered MPDAG};
  \node[reg] (CPDAG) at (1.5,0) {CPDAG};
  \node[reg] (PDAG) at (4,0) {PDAG};
  \node[reg] (invTMPDAG) at (-2,0.25) {};
  \node[reg, rotate = 35] (Edgelabel1.2) at (0.5,-1.9){(ii)};
  \node[reg, rotate = 35] (Edgelabel1.2) at (1.85,-0.95){Orients undirected};
  \node[reg, rotate = 35] (Edgelabel1.3) at (2.15,-1.15){cross-tier edges};
  \node[reg, rotate = 35] (Edgelabel2_2) at (-0.6,-0.9){(i) GES step};
  \node[reg, rotate = 35] (Edgelabel2_3) at (-0.2,-1){with $S$};
  \node[reg] (Edgelabel3.1) at (0,2.1){(iii)};
  \node[reg] (Edgelabel3.2) at (0.85,2.1){Meek's};
  \node[reg] (Edgelabel3.3) at (1.85,2.1){rule 1};

%Edges
  \draw[->, thick] (TMPDAG) -- (CPDAG);
  \draw[->, thick] (CPDAG) -- (PDAG);
  \draw [->] (PDAG) to [out=120,in=60] (invTMPDAG);
  
\end{tikzpicture}
\caption{Illustration of a step in TGES
, consisting of stages (i)-(iii). $S$ is a $\K$-consistent, $\K$-score equivalent and decomposable score criterion.} \label{fig:TGES}
\end{figure}
\\
TGES is implemented in R and available in the supplementary code.
Using Tetrad \citep{scheines1998tetrad}, it is possible to combine the algorithm Fast Greedy Equivalence Search \citep{Ramsey2017} and tiered background knowledge in order to obtain a similar but not identical algorithm. 

\subsubsection{Properties of Temporal Greedy Equivalence Search} \label{subsec:prop_TGES}
We present properties of the TGES algorithm.
\begin{lemma}
    TGES results in a tiered MPDAG.
    \label{lem:TGESResult}
\end{lemma}
The proof of Lemma \ref{lem:TGESResult} is provided in Appendix \ref{supmat:subsubsec:StageI}. 
\\
This property ensures that the graph outputted by TGES actually represents a class of DAGs and enables the use of theory based on MPDAGs, e.g. for subsequent causal inference \citep{perković2017}.
\\
This restricted equivalence class of DAGs is not guaranteed to contain the true data-generating DAG for finite samples. But as the next result shows, the restricted equivalence class contains the true data-generating DAG in the large sample limit.
\begin{theorem} \label{Thm:TGES_ALL_S_C}
(Sound and Completeness of TGES)\\
TGES using a $\K$-score equivalent, decomposable, $\K$-consistent and $\K$-locally consistent score criterion, results in a tiered MPDAG and for $n \rightarrow \infty$ the tiered MPDAG is almost surely a sound and complete estimate of the restricted Markov equivalence class of the true data-generating DAG.  
\end{theorem}
The proof and relevant definitions, lemmas, and theorems are provided in Appendix \ref{supmat:subsec:TGES_S_C}.
\section{Simulation Study} \label{sec:sim_study}
We conduct a simulation study to evaluate the strengths and limitations of TGES on finite data. We use four different metrics to measure how close the estimated graph is to the data-generating DAG. The algorithms we are testing will not be able to fully identify the data-generating DAG, as they estimate a (restricted) equivalence class. Hence, we will compare the estimated graphs to the tiered MPDAG of the data-generating DAG.

\subsection{Setup}
The simulation study is performed using R version 4.4.1. using packages pcalg \citep{pcalgRpackageGES}, causalDisco \citep{CausalDisco} and tpc \citep{TPCpackage}. The code is given as supplementary material. We have simulated 10 000 data-generating DAGs, with the number of nodes uniformly drawn from $\{7,8,...,20\}$, using an Erdős-Rényi model, with the probability of an edge being added drawn from Unif$(0.1,0.8)$ for each graph. We only use DAGs with at least one edge. We divide the nodes into three tiers with at least one node in each tier. We simulate from a multivariate Gaussian distribution with a covariance structure according to the data-generating DAG, with a sample size of $n=10000$. The causal effects of each edge are drawn from a uniform distribution between $0$ and $1$. 
\\
The algorithms GES (pcalg, \citet{pcalgRpackageGES}), STGES, TGES and TPC (tpc, \citet{TPCpackage}) are evaluated on each simulated dataset obtaining a CPDAG, PDAG, or tiered MPDAG, respectively.
GES and STGES uses Gaussian BIC and TGES uses Gaussian TBIC. TPC uses a hyperparameter $\alpha$ for conditional independence testing and we use $\alpha = 0.01$. TPC was run with two other $\alpha$ values (provided in Appendix \ref{supmat:subsec:SimStudy_AllMethods}), and $\alpha = 0.01$ was chosen since it performed the best on our metrics. 
\\
\\
First, we report the \textit{standardized Structural Hamming Distance} (sSHD), which counts how many edge reversals\footnote{A directed edge to an undirected edge and the other way around also counts as a reversal.}, additions, and removals there are between the estimated graph and the true tiered MPDAG \citep{PCALGvignette}, standardized by the number of possible edges in a graph of the given size.
\\
\\
The estimated graphs are then evaluated on recall and precision with respect to \textit{adjacencies, all directions} and \textit{in-tier directions}. 
We will not give the formal definition of these here, as details are provided in Appendix \ref{supmat:subsec:Metrics}, but we underline the fact that the direction metrics are defined only on edges where the true tiered MPDAG and the estimated graph agree on the adjacency. Therefore, we recommend taking the adjacency metric into account when interpreting the direction metrics. 
\\
The in-tier direction metric evaluates what directional information is gained, other than the trivial orientations given directly by $\K$. It is equal to evaluating the all-direction metric, but only on the in-tier edges.

\subsection{Results}
\subsubsection{Standardized Structural Hamming Distance}
\begin{figure}[!htb]
  \centering
  \includegraphics[width=0.7\linewidth]{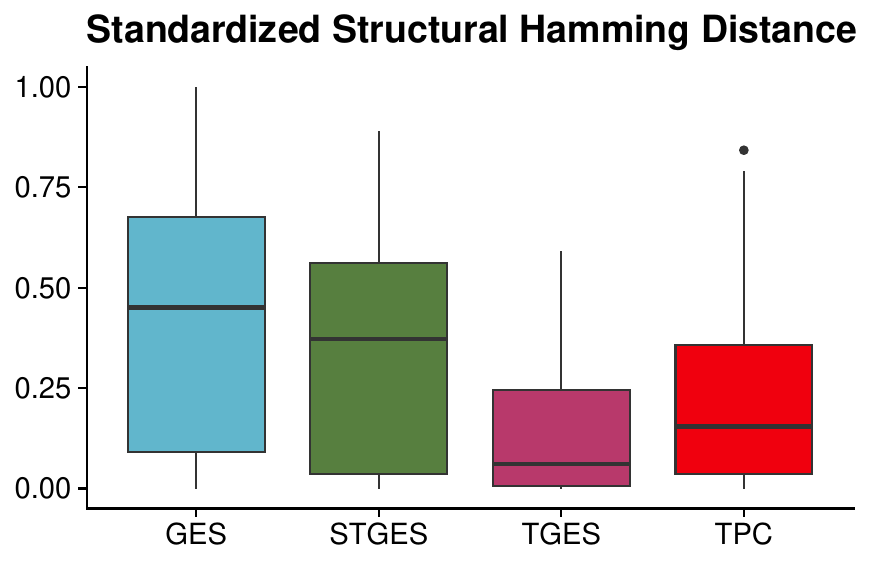}
  \caption{Boxplots showing the standardized Structural Hamming Distance of the four algorithms. The plot is based on 10 000 simulations. Smaller values correspond to better estimation.}\label{fig:Methods_sSHD}
\end{figure}
Figure \ref{fig:Methods_sSHD} presents the sSHD where higher values mean worse estimation. We see that GES performs the worst out of the four algorithms. This would be expected as we are providing tiered information to the other three algorithms and not to GES. While STGES outperforms GES, both TPC and TGES are generally fewer edge transformations from the true tiered MPDAG than both STGES and GES. In terms of sSHD, TGES is generally closer to the true graph than any other method considered.

\subsubsection{Precision and Recall}
\begin{figure}[!htb]
  \centering
  \includegraphics[width=\linewidth]{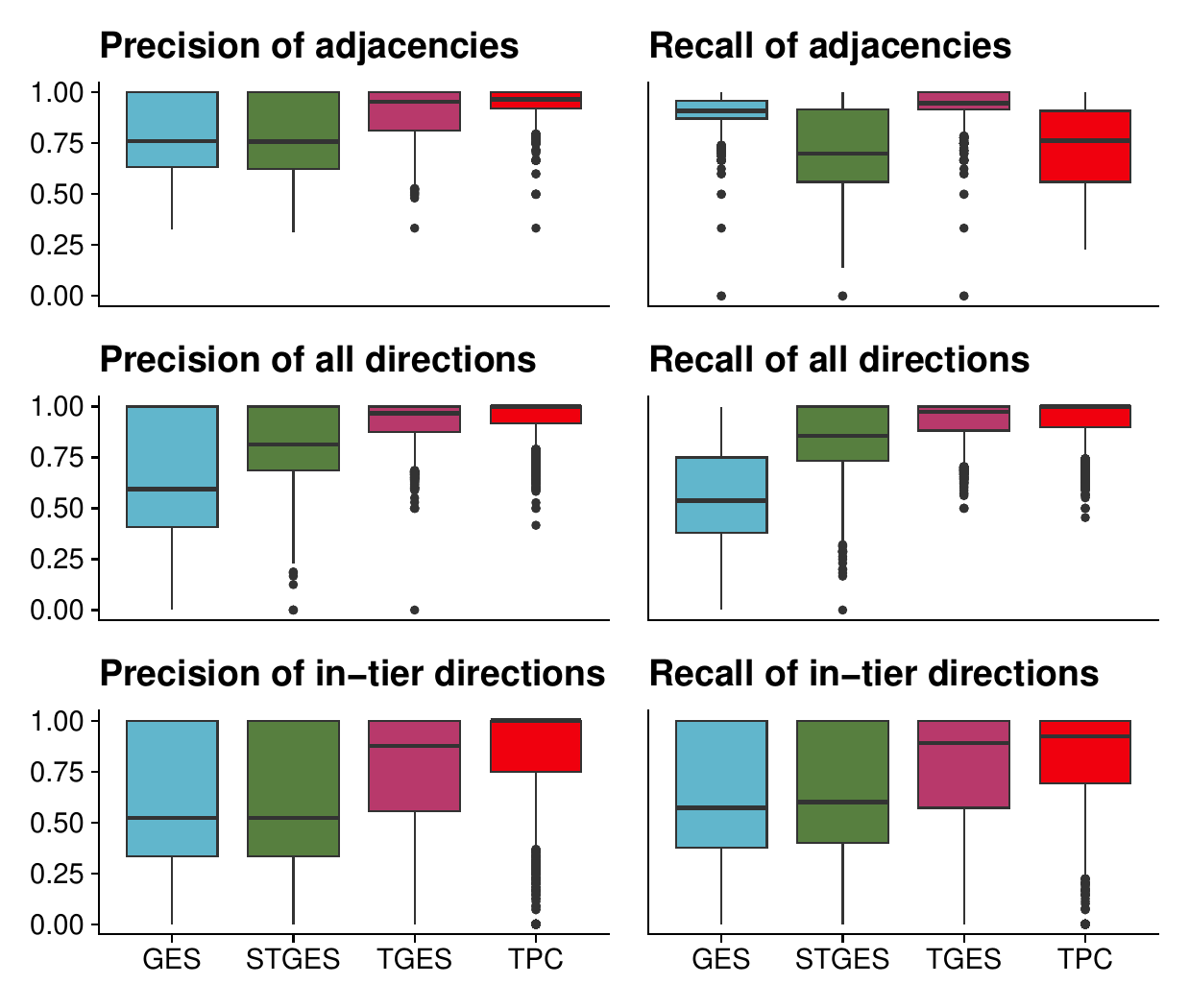}
  \caption{Boxplots showing adjacency, all direction, and in-tier direction precision and recall for the four algorithms. The plot is based on 10 000 simulations. Higher values correspond to better estimates.}\label{fig:Methods_adj_dir_int}
\end{figure}
We see in Figure \ref{fig:Methods_adj_dir_int} that TGES outperforms GES in all variants of precision and recall. While it is not so surprising that TGES orients edges better than GES, it is interesting to observe that TGES, only given additional information on directions, estimates the adjacencies with higher precision and recall as well. Tiered background knowledge does not directly infer adjacencies, only a direction given an adjacency. Hence, it is not obvious that this would result in a better estimation of the adjacencies.
\\
We notice when comparing TGES and GES in Figure \ref{fig:Methods_adj_dir_int}, that TGES is able to infer additional orientations within the tiers. This is partially by design as TGES estimates tiered MPDAGs which by construction are as or more informative than CPDAGs with regards to orientations.  
\\
\\
Next, we look at the relationship between GES and STGES in Figure \ref{fig:Methods_adj_dir_int}. We see that STGES orients edges better than GES in general, but by looking at the in-tier direction metrics we see that this is mainly due to correctly orienting the cross-tier edges, while not a lot of additional information is inferred by STGES within the tiers. Even though we saw in Figure \ref{fig:Methods_sSHD} how the sSHD was better for STGES than for GES, Figure \ref{fig:Methods_adj_dir_int} shows that STGES does a worse job of estimating adjacencies. Hence, the better sSHD can be accredited almost exclusively to the correctly oriented cross-tier edges trivially given by $\K$.
\\
\\
From the relationship between STGES and TGES, we see that embedding $\K$ in the score criterion both gives better estimation of adjacencies and better opportunities for correctly inferring in-tier edge orientations. 
\\
\\
TPC has a worse sSHD than TGES, and we see from Figure \ref{fig:Methods_adj_dir_int} that this mostly stems from the lower recall of adjacencies of TPC.
Compared to TGES, TPC has a slightly better precision of adjacencies and performs on par or better with respect to orienting edges.

\section{Data Application} \label{sec:data_application}
We apply TGES to data from The Metropolit 1953 Danish Male Birth Cohort, consisting of Danish men born in the metropolitan area of Copenhagen and includes data from their birth and up until they are approximately 65 years old \citep{CohortMetro53}. The specific dataset we use, is combined with variables from the Danish national health registers \citep{TheoryVsDataPetersen}.
\\
\\
The dataset has been used by \citet{TheoryVsDataPetersen} in a comparative study where the constraint-based temporal causal discovery algorithm TPC was compared to an expert-made DAG. Hence, we have access to a DAG made by two field experts, which we consider as the ground truth. 
\\
The data ($n=3145$) consists of 22 variables which are partitioned into 5 tiers. The variables are either continuous or binary. The expert DAG consists of 30 edges, of which 22 are cross-tier edges and 8 are in-tier edges.
\\
\\
Two tiered MPDAGs are estimated using TGES. The first one uses the standard Gaussian TBIC as a score criterion, and we will denote it by  $\text{TGES}_{\text{fixed}}$. 
\\
The second estimate, which we call $\text{TGES}_{\text{tuned}}$, uses the Gaussian TBIC with a $\lambda$ scaled penalty term \citep{ComparisonOfPackages} such that the score criterion used is:
\begin{align*}
        \TS^\lambda(\G,\D,\K) = 
\begin{cases} 
\log p(\D|\hat{\theta},\G) -\\ \lambda\frac{\#\text{parameters}}{2} \log n & \text{if } \G \text{ encodes } \K \\
-\infty & \text{if } \G \text{ contradicts } \K 
\end{cases}
\end{align*}
We choose $\lambda$ such that $\text{TGES}_{\text{tuned}}$ has the same number of edges (30) as the expert consensus DAG. 
\\
\\
Confusion matrices for the adjacency metric are provided in Tables \ref{tab:adj_conf_tges_tuned} and \ref{tab:adj_conf_tges_std}, and the expert DAG and fitted tiered MPDAGs are provided in Appendix \ref{supmat:sec:data_example}.
\\
\\
$\text{TGES}_{\text{tuned}}$ finds 10 out of the 30 adjacencies in the expert DAG and $\text{TGES}_{\text{fixed}}$ locates 3 additional correct adjacencies but also identifies 22 additional incorrect adjacencies compared to $\text{TGES}_{\text{tuned}}$. In the study performed by \citet{TheoryVsDataPetersen}, where the number of edges was tuned to 30 as well, we see the same confusion matrix performance but note that the identified edges differ (confusion matrix provided in Appendix \ref{supmat:sec:data_example}). 
\\
\\
In both estimated tiered MPDAGs, we discover a lot of edges that the experts did not identify as causal effects. While the "new" identified effects in $\text{TGES}_{\text{tuned}}$ seem possible, there are at least a few edges in $\text{TGES}_{\text{fixed}}$ which are implausible. For example, birth weight having a causal effect on whether the mother is a smoker at birth.
\\
The agreement between the experts and the two models with regards to the directions of the edges is quite high, as $\text{TGES}_{\text{tuned}}$ has 8 out of 10 correctly specified directions, and $\text{TGES}_{\text{fixed}}$ has 10 out of 12. Very few of the edges where they agree are in-tier edges (See Appendix \ref{supmat:sec:data_example}). 
\\
\\
Testing the overall adjacency performance against random guessing \citep{petersen2024doingbetterrandomguessing} results in p values $p_\text{tuned}=0.002$ and $p_\text{fixed}=0.009$. Thus, TGES finds significantly more correct adjacencies than random guessing, and the lower p-value for $\text{TGES}_{\text{tuned}}$ compared to $\text{TGES}_{\text{fixed}}$ suggests a better trade-off between additional edges and correct findings by tuning.
\begin{table}
\centering
\caption{Confusion matrix for the metric of adjacencies. Estimated using TGES and score TBIC with tuned penalty parameter to match the number of edges in the expert graph.}
\begin{tabular}{lrr}
                    & \multicolumn{2}{c}{\textbf{Expert}} \\ 
            \textbf{$\text{TGES}_{\text{tuned}}$}            & Adjacency  & Non-adjacency \\ \hline
 Adjacency       & 10         & 20            \\
 Non-adjacency & 20         & 181          
\end{tabular}
\label{tab:adj_conf_tges_tuned}
\end{table}
\begin{table}
\centering
\caption{Confusion matrix for the metric of adjacencies. Estimated using TGES and score TBIC.}
\begin{tabular}{lrr}
                    & \multicolumn{2}{c}{\textbf{Expert}} \\ 
            \textbf{$\text{TGES}_{\text{fixed}}$}            & Adjacency  & Non-adjacency \\ \hline
 Adjacency       & 13         & 42            \\
 Non-adjacency & 17         & 159          
\end{tabular}
\label{tab:adj_conf_tges_std}
\end{table}

\section{Discussion and Future Work} \label{sec:discussion}
In this paper, we introduced a temporal extension to GES, TGES, and our main result proved TGES to estimate the correct restricted equivalence class in the limit of large sample size when using a score criterion that is (locally) $\K$-consistent, decomposable, and $\K$-score equivalent. For finite sample size, we proved that TGES obtains a tiered MPDAG. Our simulation study showed that TGES is a viable, and in some cases better, alternative to the constraint-based temporal algorithm TPC. 
We demonstrated that using the tiered background knowledge in an inefficient way, like for STGES, the use case will remain theoretical. When correct temporal information is merged with statistically uncertain estimates in an unrefined way for finite samples it can lead to bad performance.
We saw from the data example, that the estimated graph by TGES is in terms of adjacencies as close to a graph agreed upon by field experts as the graph estimated by TPC, while not being the exact same graph. 
\\
\\
Throughout our applications of TGES, we only used TBIC or a modification of TBIC as a scoring criterion. Any consistent and decomposable score criterion that is finite can be extended to a temporal version by penalizing to $-\infty$ for any graph contradicting $\K$. The score criterion would then be $\K$-consistent, hence Lemma \ref{lem:TGESResult} and Theorem \ref{Thm:TGES_ALL_S_C} holds. It is thus straightforward to extend scores such as the BDeu score \citep{heckerman2015learningbayesiannetworkscombination} for discrete variables and the Conditional Gaussian BIC \citep{ScoringMixedVariables} for mixed data while maintaining the same nice properties of TGES.
\\
\\
Both $\text{TGES}_{\text{fixed}}$ and $\text{TGES}_{\text{tuned}}$ used a score criterion that assumes Gaussianity. Since some of the variables in the Metropolit data example are binary, this is misspecified \citep{Huang2018}. 
With a more refined scoring criterion, it is likely that the performance would improve.
\\
\\
We saw in our simulation study that even though we have the same desirable large sample limit properties for STGES and TGES (Theorem \ref{thm:STGES_S_C} and \ref{Thm:TGES_ALL_S_C}), the performance on finite data differs vastly. 
An analogous large sample limit property has been proved for TPC in  \citet{bang2024improvingfinitesampleperformance}, but the finite sample estimates of TPC are not generally identical to either TGES or STGES. We also saw in our real-world application that this was even true when we tuned TPC and TGES to have the same number of edges.
\\
\\
In general, we found that sSHD was lower for TGES than for TPC, but also that this was mainly due to the improved recall of adjacencies of TGES. This stresses the importance of measuring different metrics of the performance of the algorithms, and not just one metric that does not distinguish between a wrongly specified adjacency and direction. 
It would be interesting to further assess the performance by comparing with a random graph guess \citep{petersen2024doingbetterrandomguessing} or by considering a metric focused on the implications for causal effect estimation \citep{henckel2024adjustmentidentificationdistancegadjid}.
The simulation study did not stratify on true DAG density, but we expect that TPC would have a decreased performance on denser DAGs compared to TGES, like is the case for PC and GES \citep{petersen_ramsey_ekstrøm_spirtes_2023}.
\\
In our simulation study, we chose $\alpha$ for TPC post evaluations, as the goal of the simulation study was to evaluate TGES. Note that TPC thus had the advantage of being (manually) tuned to specific metrics. Hence, we expect that TGES may outperform TPC even more in real applications.
\\
\\
The simulation study indicated that by choosing to use TGES instead of TPC, we are getting slightly worse precision of adjacencies and slightly less correctly oriented in-tier edges, in return for much better recall of adjacencies. In addition to these empirical characteristics, TGES also avoids needing us to specify a hyperparameter like we have to do with $\alpha$ for TPC.
\\
TGES was shown to always estimate a tiered MPDAG. This is not a given for other algorithms such as TPC, as TPC might encounter conflicting results of tests. It is desirable to work with tiered MPDAGs, as there exist established results for these types of graphs, for example, results allowing for identification of adjustment sets and hence sub-sequent causal inference \citep{perković2017}. However, by limiting the search space to tiered MPDAGs we are possibly unnaturally forcing the result to be a tiered MPDAG. This could be an issue if the underlying causal mechanism is not acyclic. 
\\
\\
An obvious strength of using tiered background knowledge is the fact that the estimated class of DAGs becomes smaller. This results in more precise sub-sequent causal inference using for example the IDA procedure \citep{Maathuis_2009}. Additionally, tiered background knowledge is generally reliable as it often stems from a temporal structure in the data. We do however need to be cautious when assigning variables to tiers. For example, measuring the test score of a child and later measuring the genetics when the child is grown up, will not allow us to say that genetics do not have an effect on test scores.
\\
\\
Although not included in the article (See Appendix \ref{supmat:subsec:Comptime}) we found the computation time of GES and TGES to be approximately equal. This was due to TGES saving time scoring graphs in the backward phase while spending some additional time on stages (ii) and (iii).
\\
\\
All results presented here assume the background knowledge is tiered. Whether the algorithm and its results would be possible to extend for a general type of background knowledge, is not clear. While some results would naturally translate to general background knowledge, others would require us to rethink the entire structure of the algorithm. We will look into the generalization of TGES in future work.
% \\
% It would be interesting to look into how allowing for general background knowledge in TGES, would change not just the theoretical aspects but the empirical performance of the algorithm as well.
\\
\\
We have demonstrated that we have a lot to gain by utilizing temporal information. Using information about the temporal structure of data in the proper manner can have an effect on not only the estimated directions of the graph but also the adjacencies. Moreover, TGES inherits nice properties from GES: we are ensured correct estimation of the class in the limit, and we are ensured to estimate a tiered MPDAG on finite samples. TGES was also shown to be a viable alternative to the constraint-based temporal algorithm TPC on finite data.
%\begin{contributions} % will be removed in pdf for initial submission 
					  % (without ‘accepted’ option in \documentclass)
                      % so you can already fill it to test with the
                      % ‘accepted’ class option
    % Briefly list author contributions. 
    % This is a nice way of making clear who did what and to give proper credit.
    % This section is optional.

    % H.~Q.~Bovik conceived the idea and wrote the paper.
    % Coauthor One created the code.
    % Coauthor Two created the figures.
%\end{contributions}

\begin{acknowledgements} % will be removed in pdf for initial submission,
						 % (without ‘accepted’ option in \documentclass)
                         % so you can already fill it to test with the
                         % ‘accepted’ class option
    % Briefly acknowledge people and organizations here.

    % \emph{All} acknowledgements go in this section.
This work was supported by Centre for Childhood Health in collaboration with Professor Katrine Strandberg-Larsen (ID : 2024\_F\_008).
\\
\\
The authors thank K. Svalastoga, E. Høgh, P. Wolf, T. Rishøj, G. Strande-Sørensen, E. Manniche, B. Holten, I.A. Weibull and A. Ortmann, who established the Metropolit study and collected data from 1965 to 1983, and M. Osler who is the current Principal Investigator.
\end{acknowledgements}

% References
\bibliography{main}

\begin{thebibliography}{31}
\providecommand{\natexlab}[1]{#1}
\providecommand{\url}[1]{\texttt{#1}}
\expandafter\ifx\csname urlstyle\endcsname\relax
  \providecommand{\doi}[1]{doi: #1}\else
  \providecommand{\doi}{doi: \begingroup \urlstyle{rm}\Url}\fi

\bibitem[{Alain Hauser} and {Peter B\"uhlmann}(2012)]{pcalgRpackageGES}
{Alain Hauser} and {Peter B\"uhlmann}.
\newblock Characterization and greedy learning of interventional {M}arkov equivalence classes of directed acyclic graphs.
\newblock \emph{Journal of Machine Learning Research}, 13:\penalty0 2409--2464, 2012.
\newblock URL \url{https://jmlr.org/papers/v13/hauser12a.html}.

\bibitem[Andrews et~al.(2018)Andrews, Ramsey, and Cooper]{ScoringMixedVariables}
Bryan Andrews, Joseph Ramsey, and Gregory Cooper.
\newblock Scoring bayesian networks of mixed variables.
\newblock \emph{International Journal of Data Science and Analytics}, 6, 08 2018.
\newblock \doi{10.1007/s41060-017-0085-7}.

\bibitem[Bang and Didelez(2023)]{bang2023wiser}
Christine~W. Bang and Vanessa Didelez.
\newblock Do we become wiser with time? on causal equivalence with tiered background knowledge, 2023.

\bibitem[Bang et~al.(2024)Bang, Witte, Foraita, and Didelez]{bang2024improvingfinitesampleperformance}
Christine~W Bang, Janine Witte, Ronja Foraita, and Vanessa Didelez.
\newblock Improving finite sample performance of causal discovery by exploiting temporal structure, 2024.
\newblock URL \url{https://arxiv.org/abs/2406.19503}.

\bibitem[Chickering(2003)]{Chickering2002}
David~Maxwell Chickering.
\newblock Optimal structure identification with greedy search.
\newblock \emph{J. Mach. Learn. Res.}, 3\penalty0 (null):\penalty0 507–554, mar 2003.
\newblock ISSN 1532-4435.
\newblock \doi{10.1162/153244303321897717}.
\newblock URL \url{https://doi.org/10.1162/153244303321897717}.

\bibitem[Haughton(1988)]{Haughton}
Dominique M.~A. Haughton.
\newblock On the choice of a model to fit data from an exponential family.
\newblock \emph{The Annals of Statistics}, 16\penalty0 (1):\penalty0 342--355, 1988.
\newblock ISSN 00905364.
\newblock URL \url{http://www.jstor.org/stable/2241441}.

\bibitem[Hauser and B{\"u}hlmann(2012)]{hauser2012characterization}
Alain Hauser and Peter B{\"u}hlmann.
\newblock Characterization and greedy learning of interventional markov equivalence classes of directed acyclic graphs.
\newblock \emph{The Journal of Machine Learning Research}, 13\penalty0 (1):\penalty0 2409--2464, 2012.

\bibitem[Heckerman et~al.(1995)Heckerman, Geiger, and Chickering]{heckerman2015learningbayesiannetworkscombination}
David Heckerman, Dan Geiger, and David~M Chickering.
\newblock Learning bayesian networks: The combination of knowledge and statistical data.
\newblock \emph{Machine learning}, 20:\penalty0 197--243, 1995.

\bibitem[Henckel et~al.(2024)Henckel, W{\"u}rtzen, and Weichwald]{henckel2024adjustmentidentificationdistancegadjid}
Leonard Henckel, Theo W{\"u}rtzen, and Sebastian Weichwald.
\newblock Adjustment identification distance: A gadjid for causal structure learning.
\newblock In \emph{The 40th Conference on Uncertainty in Artificial Intelligence}, 2024.

\bibitem[Hernan and Robins(2024)]{hernan2020causal}
M.A. Hernan and J.M. Robins.
\newblock \emph{Causal Inference: What If}.
\newblock Chapman \& Hall/CRC Monographs on Statistics \& Applied Probab. CRC Press, 2024.
\newblock ISBN 9781420076165.
\newblock URL \url{https://books.google.dk/books?id=_KnHIAAACAAJ}.

\bibitem[Huang et~al.(2018)Huang, Zhang, Lin, Sch\"{o}lkopf, and Glymour]{Huang2018}
Biwei Huang, Kun Zhang, Yizhu Lin, Bernhard Sch\"{o}lkopf, and Clark Glymour.
\newblock Generalized score functions for causal discovery.
\newblock In \emph{Proceedings of the 24th ACM SIGKDD International Conference on Knowledge Discovery \& Data Mining}, KDD '18, page 1551–1560, New York, NY, USA, 2018. Association for Computing Machinery.
\newblock ISBN 9781450355520.
\newblock \doi{10.1145/3219819.3220104}.
\newblock URL \url{https://doi.org/10.1145/3219819.3220104}.

\bibitem[Kalisch et~al.(2024)Kalisch, Hauser, Maathuis, and Mächler]{PCALGvignette}
Markus Kalisch, Alain Hauser, Marloes~H. Maathuis, and Martin Mächler.
\newblock \emph{An Overview of the pcalg Package for R}, 2024.
\newblock URL \url{https://cran.r-project.org/package=pcalg}.
\newblock R package version 4.3.2.

\bibitem[Maathuis et~al.(2009)Maathuis, Kalisch, and Bühlmann]{Maathuis_2009}
Marloes~H. Maathuis, Markus Kalisch, and Peter Bühlmann.
\newblock Estimating high-dimensional intervention effects from observational data.
\newblock \emph{The Annals of Statistics}, 37\penalty0 (6A), December 2009.
\newblock ISSN 0090-5364.
\newblock \doi{10.1214/09-aos685}.
\newblock URL \url{http://dx.doi.org/10.1214/09-AOS685}.

\bibitem[Malinsky(2024)]{malinsky2024cautiousapproachconstraintbasedcausal}
Daniel Malinsky.
\newblock A cautious approach to constraint-based causal model selection, 2024.
\newblock URL \url{https://arxiv.org/abs/2404.18232}.

\bibitem[Meek(1995)]{meek1995causalbg}
Christopher Meek.
\newblock Causal inference and causal explanation with background knowledge.
\newblock In \emph{Proceedings of the Eleventh Conference on Uncertainty in Artificial Intelligence}, UAI'95, page 403–410, San Francisco, CA, USA, 1995. Morgan Kaufmann Publishers Inc.
\newblock ISBN 1558603859.

\bibitem[{OEIS Foundation Inc.}(2023)]{OEIS_A003024}
{OEIS Foundation Inc.}
\newblock A003024: Number of acyclic digraphs (or dags) with n labeled nodes.
\newblock \url{https://oeis.org/A003024/list}, 2023.
\newblock Accessed: 13-05-2024.

\bibitem[Osler et~al.(2005)Osler, Lund, Kriegbaum, Christensen, and Andersen]{CohortMetro53}
Merete Osler, Rikke Lund, Margit Kriegbaum, Ulla Christensen, and Anne-Marie~Nybo Andersen.
\newblock {Cohort Profile: The Metropolit 1953 Danish Male Birth Cohort}.
\newblock \emph{International Journal of Epidemiology}, 35\penalty0 (3):\penalty0 541--545, 12 2005.
\newblock ISSN 0300-5771.
\newblock \doi{10.1093/ije/dyi300}.
\newblock URL \url{https://doi.org/10.1093/ije/dyi300}.

\bibitem[Pearl(2009)]{Pearl_2009}
Judea Pearl.
\newblock \emph{Causality}.
\newblock Cambridge University Press, 2 edition, 2009.

\bibitem[Perkovic et~al.(2017)Perkovic, Kalisch, and H]{perković2017}
Emilija Perkovic, Markus Kalisch, and Maathuis~Maloes H.
\newblock Interpreting and using cpdags with background knowledge.
\newblock In \emph{The Conference on Uncertainty in Artificial Intelligence}, 2017.

\bibitem[Peters et~al.(2017)Peters, Janzing, and Schlkopf]{CausalityPeters}
Jonas Peters, Dominik Janzing, and Bernhard Schlkopf.
\newblock \emph{Elements of Causal Inference: Foundations and Learning Algorithms}.
\newblock The MIT Press, 2017.
\newblock ISBN 0262037319.

\bibitem[Petersen et~al.(2021)Petersen, Osler, and Ekstrøm]{Petersen2021datadriven}
Anne~H. Petersen, Merete Osler, and Claus Ekstrøm.
\newblock Data-driven model building for life course epidemiology.
\newblock \emph{American Journal of Epidemiology}, 190, 03 2021.
\newblock \doi{10.1093/aje/kwab087}.

\bibitem[Petersen et~al.(2023{\natexlab{a}})Petersen, Ekstr{\o}m, Spirtes, and Osler]{TheoryVsDataPetersen}
Anne~H. Petersen, Claus Ekstr{\o}m, Peter Spirtes, and Merete Osler.
\newblock Constructing causal life course models: Comparative study of data-driven and theory-driven approaches.
\newblock \emph{American Journal of Epidemiology}, 192\penalty0 (11):\penalty0 1917–1927, 2023{\natexlab{a}}.
\newblock ISSN 0002-9262.
\newblock \doi{10.1093/aje/kwad144}.

\bibitem[Petersen(2022)]{CausalDisco}
Anne~Helby Petersen.
\newblock \emph{causalDisco: Tools for Causal Discovery on Observational Data}, 2022.
\newblock URL \url{https://CRAN.R-project.org/package=causalDisco}.
\newblock R package version 0.9.1.

\bibitem[Petersen(2024)]{petersen2024doingbetterrandomguessing}
Anne~Helby Petersen.
\newblock Are you doing better than random guessing? a call for using negative controls when evaluating causal discovery algorithms, 2024.
\newblock URL \url{https://arxiv.org/abs/2412.10039}.

\bibitem[Petersen et~al.(2023{\natexlab{b}})Petersen, Ramsey, Ekstrøm, and Spirtes]{petersen_ramsey_ekstrøm_spirtes_2023}
Anne~Helby Petersen, Joseph Ramsey, Claus~Thorn Ekstrøm, and Peter Spirtes.
\newblock Causal discovery for observational sciences using supervised machine learning.
\newblock \emph{Journal of Data Science}, 21\penalty0 (2):\penalty0 255--280, 2023{\natexlab{b}}.
\newblock ISSN 1680-743X.
\newblock \doi{10.6339/23-JDS1088}.

\bibitem[Ramsey et~al.(2017)Ramsey, Glymour, Sanchez-Romero, and Glymour]{Ramsey2017}
Joseph Ramsey, Madelyn Glymour, Ruben Sanchez-Romero, and Clark Glymour.
\newblock A million variables and more: the fast greedy equivalence search algorithm for learning high-dimensional graphical causal models, with an application to functional magnetic resonance images.
\newblock \emph{International Journal of Data Science and Analytics}, 3\penalty0 (2):\penalty0 121--129, Mar 2017.
\newblock ISSN 2364-4168.
\newblock \doi{10.1007/s41060-016-0032-z}.
\newblock URL \url{https://doi.org/10.1007/s41060-016-0032-z}.

\bibitem[Ramsey and Andrews(2017)]{ComparisonOfPackages}
Joseph~D. Ramsey and Bryan Andrews.
\newblock A comparison of public causal search packages on linear, gaussian data with no latent variables.
\newblock \emph{CoRR}, abs/1709.04240, 2017.
\newblock URL \url{http://arxiv.org/abs/1709.04240}.

\bibitem[Scheines et~al.(1998)Scheines, Spirtes, Glymour, Meek, and Richardson]{scheines1998tetrad}
Richard Scheines, Peter Spirtes, Clark Glymour, Christopher Meek, and Thomas Richardson.
\newblock The tetrad project: Constraint based aids to causal model specification.
\newblock \emph{Multivariate Behavioral Research}, 33\penalty0 (1):\penalty0 65--117, 1998.

\bibitem[Schwarz(1978)]{BICSchwarz}
Gideon Schwarz.
\newblock {Estimating the Dimension of a Model}.
\newblock \emph{The Annals of Statistics}, 6\penalty0 (2):\penalty0 461 -- 464, 1978.
\newblock \doi{10.1214/aos/1176344136}.
\newblock URL \url{https://doi.org/10.1214/aos/1176344136}.

\bibitem[Spirtes and Glymour(1991)]{PC_article}
Peter Spirtes and Clark Glymour.
\newblock An algorithm for fast recovery of sparse causal graphs.
\newblock \emph{Social Science Computer Review - SOC SCI COMPUT REV}, 9:\penalty0 62--72, 04 1991.
\newblock \doi{10.1177/089443939100900106}.

\bibitem[Witte(2023)]{TPCpackage}
Janine Witte.
\newblock \emph{tpc: Tiered PC Algorithm}, 2023.
\newblock URL \url{https://CRAN.R-project.org/package=tpc}.
\newblock R package version 1.0.

\end{thebibliography}

\newpage

\onecolumn
\title{Score-Based Causal Discovery with Temporal Background Information\\(Appendix)}
\maketitle

\appendix
\section{Algorithms} \label{supmat:sec:alg}
\subsection{Simple Temporal Greedy Equivalence Search}

\begin{algorithm}[H]
\caption{Simple Temporal Greedy Equivalence Search (STGES)}
\label{alg:STGES}
\begin{algorithmic}[1]
\Statex \textbf{Input:} $\D, \K$ 
\State $\mathcal{W} \gets$ CPDAG from Greedy Equivalence Search on $\D$
\State $\mathcal{W} \gets$ Remove directed edges in $\mathcal{W}$ contradicting $\K$
\State $\mathcal{W} \gets $ Use Algorithm \ref{alg:restrict} on $\mathcal{W}$
\State $\mathcal{W} \gets$ Iterate over Meek's rules 1-4 on $\mathcal{W}$ 
\State \textbf{Output:} 
$\mathcal{W}$
\end{algorithmic}
\end{algorithm}

\begin{figure*}[h]
\centering
\begin{tikzpicture}[>=Stealth, node distance=2cm, every node/.style={circle, draw, thick}]
  \tikzset{reg/.style={draw = none,font=\sffamily}}
  \tikzset{sqr/.style={
  draw, % Enables drawing the border. By default, the border is rectangular.
  font=\sffamily, % Keeps the sans-serif font
  rectangle, % Explicitly defines the shape as a rectangle
  minimum size=1cm, % Ensures that the node is at least a square of 1cm by 1cm
  align=center % Aligns the text in the center of the node
}}
%Nodes
  \node[reg] (GES) at (0,0) {CPDAG};
  \node[sqr] (HiddenGES) at (-3,0) {Data};
  \node[reg] (PDAG1) at (3,0) {PDAG};
  \node[reg] (PDAG2) at (6,0) {PDAG};
  \node[reg] (TMPDAG) at (9.5,0) {$\text{Tiered}$ $\text{MPDAG}$};
%Edges
  \draw[->, thick] (GES) -- (PDAG1);
  \draw[->, thick] (HiddenGES) -- (GES);
  \draw[->, thick] (PDAG1) -- (PDAG2);
  \draw[->, thick] (PDAG2) -- (TMPDAG);
  \node[reg, rotate = 45] (Edgelabel1) at (0.3,-1.5){Remove directed edges};
  \node[reg, rotate = 45] (Edgelabel2) at (0.8,-1.6){contradicting $\K$};
  \node[reg, rotate = 45] (Edgelabel3) at (4,-1.2){Algorithm};
  \node[reg, rotate = 45] (Edgelabel3) at (4.6,-0.58){\ref{alg:restrict}};
  \node[reg, rotate = 45] (Edgelabel4) at (6.75,-1.2){Meek's rules};
  \node[reg, rotate = 45] (Edgelabel0) at (-1.8,-0.7){GES}; 
\end{tikzpicture}
\caption{An overview of the steps in Algorithm \ref{alg:STGES} (STGES).}
\label{fig:SimTGES}
\end{figure*}
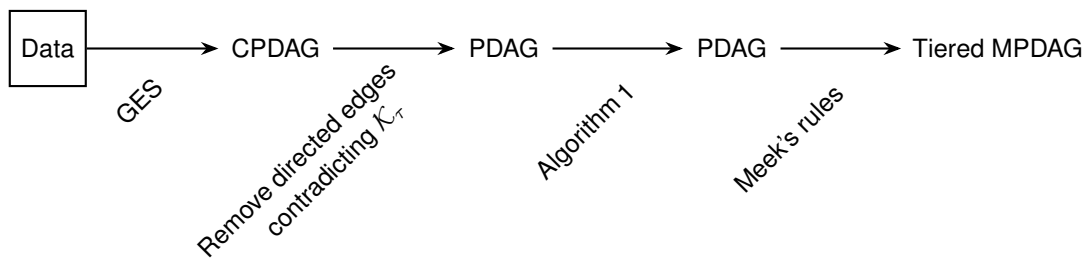

\subsection{Temporal Greedy Equivalence Search}
In this subsection, we provide the algorithms for each type of step (forward, backward, and turning) in TGES (see Algorithm \ref{alg:TGES}). 
\\
\\
We will refer to an edge addition as \textit{valid}, as defined in Theorem 15 of \citet{Chickering2002}, and an edge removal as valid, as defined in Theorem 17 of \citet{Chickering2002}.
\\
\\
Let $S$ be a $\K$-consistent $\K$-score equivalent and decomposable score criterion with local score $s$.

\begin{algorithm}[H]
\caption{Temporal\_Forward\_Step}
\label{alg:T_for_ph}
\begin{algorithmic}[1]
\Statex \textbf{Input:} $\mathcal{W},\D,\K$
\Statex \textbf{Stage (i):}
\For{all pair of nodes $(X,Y)$ in $\mathcal{W}$ where edge addition is valid}
\State $\Delta s_{X,Y} \gets s(X_i,Pa^\mathcal{W}_i \cup Y,\K) - s(X_i,Pa^\mathcal{W}_i,\K)$
\EndFor
\If{ any $\Delta s_{X,Y} > 0$}
\State $(X',Y') \gets \underset{(X,Y)}{\argmax} \left( \Delta s_{X,Y}\right)$
\State Insert edge $\left\{Y' \rightarrow X'\right\}$ in $\mathcal{W}$
\State $\mathcal{W} \gets$ CPDAG of $\mathcal{W}$ 
\Else
\State Output $\mathcal{W}$ with no changes made
\EndIf
\Statex \textbf{Stage (ii):} 
\State $\mathcal{W} \gets$ Use Algorithm \ref{alg:restrict} on $\mathcal{W}$
\Statex \textbf{Stage (iii):}
\State $\mathcal{W} \gets$ Iterate over Meek's rule 1 on $\mathcal{W}$
\State \textbf{Output:} $\mathcal{W}$
\end{algorithmic}
\end{algorithm}

\begin{algorithm}[H]
\caption{Temporal\_Backward\_Step}
\label{alg:T_bac_ph}
\begin{algorithmic}[1]
\Statex \textbf{Input:} $\mathcal{W},\D,\K$
\Statex \textbf{Stage (i):}
\For{all pair of nodes $(X,Y)$ in $\mathcal{W}$ where edge removal is valid}
\State $\Delta s_{X,Y} \gets s(X_i,Pa^\mathcal{W}_i \setminus Y,\K) - s(X_i,Pa^\mathcal{W}_i,\K)$
\EndFor
\If{ any $\Delta s_{X,Y} > 0$}
\State $(X',Y') \gets \underset{(X,Y)}{\argmax} \left( \Delta s_{X,Y}\right)$
\State Remove edge $\left\{Y' \rightarrow X'\right\}$ in $\mathcal{W}$
\State $\mathcal{W} \gets$ CPDAG of $\mathcal{W}$ 
\Else
\State Output $\mathcal{W}$ with no changes made
\EndIf
\Statex \textbf{Stage (ii):} 
\State $\mathcal{W} \gets$ Use Algorithm \ref{alg:restrict} on $\mathcal{W}$
\Statex \textbf{Stage (iii):}
\State $\mathcal{W} \gets$ Iterate over Meek's rule 1 on $\mathcal{W}$
\State \textbf{Output:} $\mathcal{W}$
\end{algorithmic}
\end{algorithm}

\begin{algorithm}[H]
\caption{Temporal\_Turning\_Step}
\label{alg:T_tur_ph}
\begin{algorithmic}[1]
\Statex \textbf{Input:} $\mathcal{W},\D,\K$
\Statex \textbf{Stage (i):}
\For{all directed edges$\left\{X \rightarrow Y\right\}$ in $\mathcal{W}$}
\State $\Delta s_{X,Y} \gets s(X_i,Pa^\mathcal{W}_{X_i} \setminus Y,\K) - s(X_i,Pa^\mathcal{W}_{X_i},\K) + s(Y_i,Pa^\mathcal{W}_{Y_i} \cup X,\K) - s(Y_i,Pa^\mathcal{W}_{Y_i},\K)$ 
\EndFor
\If{ any $\Delta s_{X,Y} > 0$}
\State $(X',Y') \gets \underset{(X,Y)}{\argmax} \left( \Delta s_{X,Y}\right)$
\State Reverse edge $\left\{X' \rightarrow Y'\right\}$ to $\left\{Y' \rightarrow X'\right\}$ in $\mathcal{W}$
\State $\mathcal{W} \gets$ CPDAG of $\mathcal{W}$ 
\Else
\State Output $\mathcal{W}$ with no changes made
\EndIf
\Statex \textbf{Stage (ii):} 
\State $\mathcal{W} \gets$ Use Algorithm \ref{alg:restrict} on $\mathcal{W}$
\Statex \textbf{Stage (iii):}
\State $\mathcal{W} \gets$ Iterate over Meek's rule 1 on $\mathcal{W}$
\State \textbf{Output:} $\mathcal{W}$
\end{algorithmic}
\end{algorithm}

\section{Proofs and Definitions} \label{supmat:sec:proof}
\subsection{Meek's Rules} \label{supmat:subsec:Meeks}
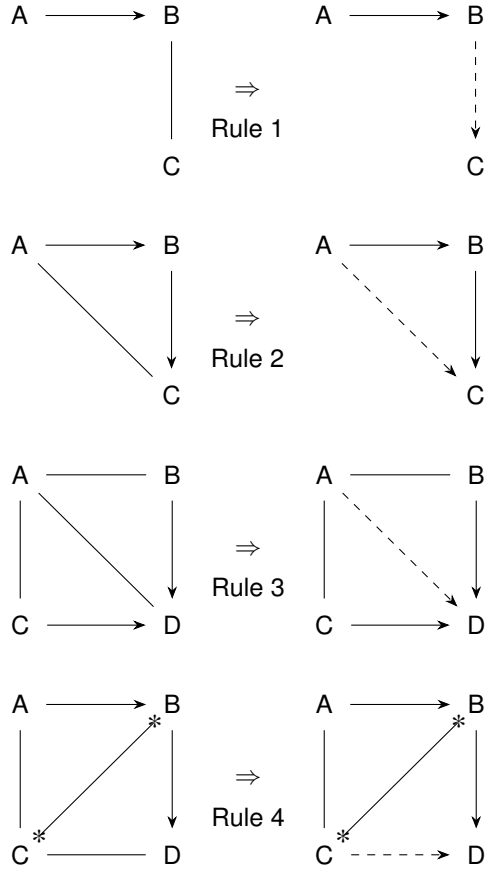
\begin{figure}[H]
\begin{center}
\begin{tikzpicture}[>=Stealth, node distance=2cm, every node/.style={circle, draw, thick}]

  \tikzset{reg/.style={draw = none,font=\sffamily}}

%Nodes
  \node[reg] (A) {A};
  \node[reg] (B) [right of=A] {B};
  \node[reg] (C) [below of=B] {C};
  \node[reg] (label1) at (3,-1) {$\Rightarrow$};
  \node[reg] (label2) at (3,-1.5) {Rule 1};
  
  \node[reg] (A2) [right of=B] {A};
  \node[reg] (B2) [right of=A2] {B};
  \node[reg] (C2) [below of=B2] {C};
  
%Edges
  \draw[->] (A) -- (B);
  \draw[-] (B) -- (C);
  \draw[->] (A2) -- (B2);
  \draw[->, dashed] (B2) -- (C2);
\end{tikzpicture}
\end{center}
\begin{center}
\begin{tikzpicture}[>=Stealth, node distance=2cm, every node/.style={circle, draw, thick}]

  \tikzset{reg/.style={draw = none,font=\sffamily}}

%Nodes
  \node[reg] (A) {A};
  \node[reg] (B) [right of=A] {B};
  \node[reg] (C) [below of=B] {C};
  \node[reg] (label1) at (3,-1) {$\Rightarrow$};
  \node[reg] (label2) at (3,-1.5) {Rule 2};
  
  \node[reg] (A2) [right of=B] {A};
  \node[reg] (B2) [right of=A2] {B};
  \node[reg] (C2) [below of=B2] {C};
  
%Edges
  \draw[->] (A) -- (B);
  \draw[->] (B) -- (C);
  \draw[->] (A2) -- (B2);
  \draw[->] (B2) -- (C2);
  \draw[-] (A) -- (C);
  \draw[->,dashed] (A2) -- (C2);
\end{tikzpicture}
\end{center}
\begin{center}
\begin{tikzpicture}[>=Stealth, node distance=2cm, every node/.style={circle, draw, thick}]

  \tikzset{reg/.style={draw = none,font=\sffamily}}

%Nodes
  \node[reg] (A) {A};
  \node[reg] (B) [right of=A] {B};
  \node[reg] (D) [below of=B] {D};
  \node[reg] (C) [below of=A] {C};
  \node[reg] (label1) at (3,-1) {$\Rightarrow$};
  \node[reg] (label2) at (3,-1.5) {Rule 3};
  
  \node[reg] (A2) [right of=B] {A};
  \node[reg] (B2) [right of=A2] {B};
  \node[reg] (D2) [below of=B2] {D};
  \node[reg] (C2) [below of=A2] {C};
  
%Edges
  \draw[-] (A) -- (B);
  \draw[->] (B) -- (D);
  \draw[-] (A2) -- (B2);
  \draw[->] (B2) -- (D2);
  \draw[-] (A) -- (C);
  \draw[-] (A2) -- (C2);
  \draw[->] (C) -- (D);
  \draw[-] (A) -- (D);
  \draw[->, dashed] (A2) -- (D2);
  \draw[->] (C2) -- (D2);
  
\end{tikzpicture}
\end{center}
\begin{center}
\begin{tikzpicture}[>=Stealth, node distance=2cm, every node/.style={circle, draw, thick}]

  \tikzset{reg/.style={draw = none,font=\sffamily}}

%Nodes
  \node[reg] (A) {A};
  \node[reg] (B) [right of=A] {B};
  \node[reg] (D) [below of=B] {D};
  \node[reg] (C) [below of=A] {C};
  \node[reg] (label1) at (3,-1) {$\Rightarrow$};
  \node[reg] (label2) at (3,-1.5) {Rule 4};
  
  \node[reg] (A2) [right of=B] {A};
  \node[reg] (B2) [right of=A2] {B};
  \node[reg] (D2) [below of=B2] {D};
  \node[reg] (C2) [below of=A2] {C};
  \node[reg,font=\large] (star1) at (0.25,-1.8) {*};
  \node[reg,font=\large] (star2) at (1.77,-0.275) {*};
  \node[reg,font=\large] (star3) at (4.25,-1.8) {*};
  \node[reg,font=\large] (star4) at (5.77,-0.275) {*};
%Edges
  \draw[->] (A) -- (B);
  \draw[->] (B) -- (D);
  \draw[->] (A2) -- (B2);
  \draw[->] (B2) -- (D2);
  \draw[-] (A) -- (C);
  \draw[-] (A2) -- (C2);
  \draw[-] (C) -- (D);
  \draw[-] (B) -- (C); 
  \draw[-] (B2) -- (C2);
  \draw[->,dashed] (C2) -- (D2);

\end{tikzpicture}
\end{center}
\caption{Representations of Meek's rules 1-4 \citep{meek1995causalbg}. A dashed edge is inferred by the rule. The symbol $*$ is used to denote that the edge between $B$ and $C$ could either be $\{B \rightarrow C\}$, $\{B \leftarrow C\}$ or $\{B \text{ --- } C\}$.}
\label{fig:Meeksrules1to4}
\end{figure}
\subsection{Bayesian Information Criterion}\label{supmat:subsec:BIC}
\begin{definition}(Decomposable score)
    \\
    A score criterion $S$ is decomposable if we can write it as a sum of local scores of the nodes given the parents of the nodes
    \begin{align}
        S(\G,\D) = \sum^d_{i=1}s(X_i,Pa^\G_i),
        \label{eq:decomp}
    \end{align}
where $s$ denotes the local score function, $X_i$ is a node in $\G$, $\text{Pa}_i^{\G}$ are the parents of $X_i$ in $\G$.
\end{definition}
We have that local BIC, denoted $s_B$, is given as
\begin{align}
\label{eq:localBIC}
    S_B(\G,\D) = \sum^d_{i=1}s_B(X_i,Pa^\G_i) = \sum^d_{i=1}\log p(\D_i|\hat{\theta},\G) - \frac{\#\text{parameters}}{2}\log m\text{,}
\end{align}
where $\D_i$ is the data of the nodes $X_i$ and $Pa^\G_i$. $p(\D_i|\hat{\theta},\G)$ is the likelihood of the local model defined on $X_i$ and $Pa^\G_i$ and $\#\text{parameters}$ is the number of parameters in this local model. In the Gaussian case $\#\text{parameters} = \#\text{edges}$.
\\
\\
We will denote \textit{d-separation} in a DAG $\G$ as $\indG$, and no d-separation in DAG $\G$ as $\nindG$. 
\\
\\
Next, we introduce some graph notation from \citet{Chickering2002} concentrating on independence statements:
\\ 
Let $\G$ and $\HH$ be two DAGs defined on the same set of nodes $\V$. We write $\G\leq\mathcal{H}$ when every independence relationship in $\mathcal{H}$ holds in $\G$: For all disjoint node sets $\mathbf{A},\mathbf{B},\mathbf{C} \subseteq \V$
\begin{align*}
     \mathbf{A} \perp\!\!\!\perp_{\HH} \mathbf{B} | \mathbf{C} \Rightarrow \mathbf{A} \indG \mathbf{B} | \mathbf{C}
\end{align*}
The notation $\leq$ is motivated as follows: if $\G\leq\mathcal{H}$ then $\G$ has fewer edges than $\HH$. If there exists an independence statement in $\HH$ which is not in $\G$, we write $\G \not\leq \HH$.
\\
If we have that $\G\leq\mathcal{H}$ and $\HH\leq\G$ then we write $\G \approx \HH$. If $\G \approx \HH$ then the two DAGs are Markov equivalent and hence belong to the same Markov equivalence class. In particular, if $\G \approx \HH$ and they both encode the same tiered background knowledge $\K$ they belong to the same restricted equivalence class.
\begin{definition}
    (Consistent Scoring Criterion)\\
Let $\GS$ be the data-generating DAG from which the data $ \D $ were generated and $ n $ be the number of observations. A scoring criterion $ S $ is consistent if for $n \rightarrow \infty$, the following two properties hold almost surely:
\begin{enumerate}
    \item[i.] If $\GS \leq \G'$ and $\GS \not\leq \G$, then $ S(\G', \D) > S(\G, \D) $.
    \item[ii.] If $\GS \leq \G$, $\GS \leq \G'$, and $\G'$ has fewer edges than $\G$, then $ S(\G', \D) > S(\G, \D) $ .
\end{enumerate}
\end{definition}
\begin{definition}
    (Locally Consistent Scoring Criterion)\\
Let $\D$ be a data set consisting of $ n $ records that are i.i.d. samples from some distribution $p$. Let $\G$ be any DAG with no edge $\left\{X \rightarrow Y\right\}$, and let $\G'$ be the DAG that results from adding the edge $\left\{X \rightarrow Y\right\}$ to $\G$. A scoring criterion $S(\G, \D)$ is locally consistent if as $n\rightarrow \infty$, the following two properties hold almost surely:
\begin{enumerate}
    \item[i.] If $ Y \not\!\perp\!\!\!\perp_p X | Pa^{\G}_Y $, then $ S(\G', \D) > S(\G, \D) $.
    \item[ii.] If $ Y \perp\!\!\!\perp_p X | Pa^{\G}_Y $, then $ S(\G', \D) < S(\G, \D) $.
\end{enumerate}
\end{definition}
\begin{definition}(Score equivalent)\\
We will say a scoring criterion $S$ is score equivalent if for any pair of DAGs $\G,\G'$ in a given equivalence class $\E$
\begin{align*}
    S(\G,\D)=S(\G',\D)
\end{align*}
\end{definition}
\begin{proposition}
\label{prop:BICallProps}
    BIC is decomposable, consistent, locally consistent, and score equivalent on i.i.d. data from a multivariate Gaussian distribution or multinomial distributions for discrete variables.
\end{proposition}
\begin{proof}
    BIC is consistent by \citet{Haughton}, locally consistent by Lemma 7 from \citet{Chickering2002}, decomposable and score equivalent on i.i.d. data from a multivariate Gaussian distribution or multinomial distributions for discrete variables by \citet{Chickering2002}.
\end{proof}
\subsection{Simple Temporal Greedy Equivalence Search} \label{supmat:subsec:STGES}
This subsection contains the proof of Theorem \ref{thm:STGES_S_C}. We restate the Theorem:
\begin{customtheorem}{\ref{thm:STGES_S_C}}
    (Sound and Completeness of STGES)\\ 
    For $n \rightarrow \infty$ STGES almost surely results in a tiered MPDAG which is a sound and complete estimate of the restricted Markov equivalence class of the true data-generating DAG    
\end{customtheorem}

\begin{proof}
Assume that $\E$ is the equivalence class that is a result of running GES to completion for $n \rightarrow \infty$.
By Theorem \ref{thm:GES_S_C} we have that the conditional independence statements in the true data-generating DAG $\GS$ are the same as the conditional independence statements in $\E$. Therefore it must be that $\GS \in \E$.  
Since the data-generating DAG $\GS$ encodes $\K$, and we have that $\GS \in \E$, we know that $\E$ is in agreement with $\K$. Hence, we are not removing any edges in "Remove directed edges contradicting $\K$" in Algorithm \ref{alg:STGES}, and the graph resulting from Algorithm \ref{alg:restrict} and Meek's rules is a tiered MPDAG which have the same conditional independence statements as $\E$. Hence, the estimated tiered MPDAG represents the true restricted Markov equivalence class and the estimate is sound and complete.
\end{proof}

\subsection{Temporal Bayesian Information Criterion}\label{supmat:subsec:TBIC}
We will throughout this subsection assume, data $\D$ is i.i.d. multivariate Gaussian.
\\
\\
We will show that the Temporal Bayesian Information Criterion is decomposable into a sum of local scores. In order to prove this we define the local score function of TBIC.%, $\Ts$
\begin{align}
    \label{eq:localTBIC}
     \Ts(X_i,Pa^\G_i,\K) = 
\begin{cases} 
s_B(X_i,Pa^\G_i) & \text{if } \forall  A \in Pa^\G_i \text{ we have that } \tau(X_i) \geq \tau(A) \\
-\infty & \text{if } \exists A \in Pa^\G_i \text{ such that } \tau(X_i) < \tau(A)
\end{cases}
\end{align}
Here, $s_B$ is the local score of BIC from (\ref{eq:localBIC}) and $X_i$ is a node in $\G$.
\\
\begin{proposition}
    (TBIC is decomposable) \label{prop:TBICDecomp}
    \\
    We can write $\TS$ as a sum of local scores $\Ts$ of nodes given their parents, with $\Ts$ given in (\ref{eq:localTBIC}):
    \begin{align}
        \TS(\G,\D,\K) = \sum^d_{i=1}\Ts(X_i,Pa^\G_i,\K)
    \end{align}
\end{proposition}
\begin{proof}
    We separate this proof into two cases: when $\G$ encodes $\K$ and when $\G$ contradicts $\K$.
    \\
    \\
    First let us assume that $\G = (\V,\mathbf{E})$ encodes $\K$. Then for all edges $\{A \rightarrow B\} \in \mathbf{E}$ we have that $\{A \rightarrow B\} \notin \mathcal{F}$, where $\mathcal{F}$ is the set of forbidden edges from $\K$. Hence we have that $\tau(A) \leq \tau(B)$ for all $\{A \rightarrow B\} \in \mathbf{E}$, which is equivalent to stating that for all $X_i \in \V$ it must be that $\forall A \in Pa^\G_i$ have that $\tau(A) \leq \tau(X_i)$. Thus we have that $\Ts(X_i,Pa^\G_i,\K) = s_B(X_i,Pa^\G_i)$.
    \\
    Then 
    \begin{align*}
        \TS(\G,\D,\K) = S_B(\G,\D) = \sum^d_{i=1}s_B(X_i,Pa^\G_i) = \sum^d_{i=1}\Ts(X_i,Pa^\G_i,\K)\text{.}
    \end{align*}
    Now let us instead assume that $\G$ contradicts $\K$. Then there exists an edge $  \{A \rightarrow B\} \in \mathbf{E}$ such that $\{A \rightarrow B\} \in \mathcal{F}$. Hence there exists an edge $ \{A \rightarrow B\} \in \mathbf{E}$ such that $\tau(A) > \tau(B)$, which is equivalent to stating that there exist an $X_q \in \V$ with a parent $A \in Pa^{\G}_q$ such that $\tau(A) > \tau(X_q)$. Thus we have for $X_q$ that $\Ts(X_q,Pa^\G_q,\K)=-\infty$.
    \\
    Then,
    \begin{align*}
        &\sum^d_{i=1}\Ts(X_i,Pa^\G_i,\K) = \sum^{q-1}_{i=1}\Ts(i) + \sum^d_{i=q+1}\Ts(i) +  \Ts(X_q,Pa^\G_q,\K) 
        \\=& \sum^{q-1}_{i=1}\Ts(i) + \sum^d_{i=q+1}\Ts(i) - \infty \overset{\dagger}{=} -\infty = \TS(\G,\D,\K)\text{,}
    \end{align*}
    where we shortened $\Ts(X_i,Pa^\G_i,\K)$ as $\Ts(i)$ for easier notation. 
    $\dagger$ is true since all $\Ts(X_i,Pa^\G_i,\K)<\infty$.
\end{proof}
\begin{definition}($\K$-Consistency)\\
  Let $\GS$ be the DAG from which the data $ \D $ were generated, $ n $ be the number of observations, and $\K$ tiered background knowledge. A scoring criterion $S$ is $\K$-consistent if the following all hold: \label{def:K_Cons} 
   \begin{enumerate}

      \item[i.] If $\G'$, $\G$ both encode $\K$ and we have that $\GS \leq \G'$ and $\GS \not\leq \G$, then for $n \rightarrow \infty $ $ S(\G', \D) > S(\G, \D) $ almost surely.
      
      \item[ii.] If $\G'$, $\G$ both encode $\K$, $\GS \leq \G$, $\GS \leq \G'$, and $\G'$ has fewer edges than $\G$, then for $n \rightarrow \infty $ $ S(\G', \D) > S(\G, \D) $  almost surely.

      \item[iii.] If $\G'$ encodes $\K$ and $\G$ contradicts $\K$ then $S(\G',\D,\K) > S(\G, \D,\K)$.
   \end{enumerate}
\end{definition}
\begin{proposition} \label{prop:TBICcons}
    TBIC is a $\K$-consistent scoring criterion 
\end{proposition}
\begin{proof}
    $(i)$ \& $(ii)$: If $\G$ and $\G'$ encode the tiered background knowledge $\K$ then the scoring criterion TBIC is equal to the scoring criterion $BIC$ which by Proposition \ref{prop:BICallProps} has property $(i)$ and $(ii)$.
    \\
    $(iii)$: Assume $\G'$ encodes $\K$ and $\G$ contradicts $\K$. Then the property follows from the fact that  $\TS(\G',\D,\K) = S_B(\G',\D) > -\infty$ and $\TS(\G, \D,\K) = -\infty$.  
\end{proof}
\begin{proposition} \label{prop:sc_loc_cons}
    A $\K$-consistent and decomposable scoring criterion is a locally $\K$-consistent scoring criterion.
\end{proposition}
\begin{proof}
    Let $\G$ be any DAG that encodes $\K$ and where adding the edge $\left\{X \rightarrow Y\right\}$ to $\G$ results in a DAG $\G'$.
    \\
    It must then be that $X\in\ndeG_Y$ since adding the edge $\left\{X\rightarrow Y\right\}$ results in a DAG and therefore no cycle. Hence, $Y \indG X |\paG_Y$.
    \\
    Assume $S$ is a decomposable and $\K$-consistent score criterion. We show that each of the properties $(i)-(iii)$ for a locally $\K$-consistent score criterion is held by $S$.
    \\
    \\
    $(i)$ \& $(ii)$: Assume $\tau(X)\leq \tau(Y)$. By decomposability of the score criterion, the increase in score that results from adding $\left\{X \rightarrow Y\right\}$ to $\G$ is the same increase in score that results from adding $\left\{X \rightarrow Y\right\}$ to any other DAG $\HH$ where $\text{Pa}_{Y}^{\HH}= \paG_Y$. Choose $\HH$ as a DAG where $\text{Pa}_{Y}^{\HH}= \paG_Y$ and where adding $\left\{X \rightarrow Y\right\}$ results in a completely connected DAG $\HH'$ that encodes $\K$. Then $\HH$ will by default also encode $\K$. Since $\HH'$ is fully connected, we have that $\GS\leq\HH'$. We now continue the proof for $(i)$ and $(ii)$ separately.
    \begin{adjustwidth}{0.5cm}{}
    $(i)$: Assume  $ Y \not\!\perp\!\!\!\perp_p X | Pa^{\G}_Y $, then  $ Y \not\!\perp\!\!\!\perp_p X | Pa^{\HH}_Y $ since $\text{Pa}_{Y}^{\HH}= \paG_Y$. Since the data-generating DAG is Markov and faithful to the true distribution $p$, we have that $ Y \nindGS X | Pa^{\HH}_Y $. From this, we have that there exists an independence statement in $\HH$ which is not in $\GS$ and therefore $\GS \not\leq \HH$. By $\K$-consistency ($(i)$ in Def. \ref{def:K_Cons}) we have that $S(\HH,\D,\K) < S(\HH',\D,\K) $ which implies that $S(\G,\D,\K) < S(\G',\D,\K) $.    
    \\
    \\
    $(ii)$: First let us realize that the only d-separation statement in $\HH$ is $ Y \!\perp\!\!\!\perp_p X | Pa^{\HH}_Y $. One less node in the conditioning set would open the path through a parent of $Y$. Since all nodes except $X$ and $Y$ are adjacent in $\HH$, adding an extra node $Z$ to the conditioning set would imply $Z\in\text{Ch}_Y^{\HH}$. $Z$ is also adjacent to $X$ and hence either $Z\in\text{Pa}_X^{\HH}$ or $Z\in\text{Ch}_X^{\HH}$. If $Z\in\text{Ch}_X^{\HH}$, then the path $X\rightarrow Z\leftarrow Y$ is open when conditioning on $Z$ thus there is no independence. If $Z\in\text{Pa}_X^{\HH}$ then $Y\in \text{An}_X^{\HH}$ and adding the edge $\{X \rightarrow Y\}$ would introduce a cycle, and contradict that $\HH'$ is a DAG.
    \\
    Assume  $ Y \!\perp\!\!\!\perp_p X | Pa^{\G}_Y $, then  $ Y \!\perp\!\!\!\perp_p X | Pa^{\HH}_Y $ which implies $ Y \indGS X | Pa^{\HH}_Y $. As this is the only independence statement in $\HH$, we have that $\GS\leq \HH$. By $\K$-consistency ($(ii)$ in Def. \ref{def:K_Cons}) we have that $S(\HH,\D,\K) > S(\HH',\D,\K) $ which implies that $S(\G,\D,\K) > S(\G',\D,\K) $.
    \end{adjustwidth}
    $(iii)$: Assume $\tau(X) > \tau(Y)$. Then $\G'$ contradicts $\K$, and by $\K$-consistency ($(iii)$ in Def. \ref{def:K_Cons}) we have that $S(\G,\D,\K) > S(\G',\D,\K) $.
\end{proof}
\begin{proposition} \label{prop:TBICloc_cons}
    TBIC is a locally $\K$-consistent scoring criterion.
\end{proposition}
\begin{proof}
    TBIC is $\K$-consistent and decomposable by Proposition \ref{prop:TBICDecomp} and Proposition \ref{prop:TBICcons}. It then follows from Proposition \ref{prop:sc_loc_cons} that TBIC is a locally $\K$-consistent scoring criterion.
\end{proof}
\begin{proposition} (TBIC is $\K$-score equivalent)\label{prop:TBICScoreEq}
    \\ 
    For any DAGs $\G,\G'\in \EK$, we have that the TBIC scores are identical  $\TS(\G,\D,\K)=\TS(\G',\D,\K)$.
\end{proposition}
\begin{proof}
    Both $\G$ and $\G'$ encode $\K$ hence $\TS(\G,\D,\K)=S_B(\G,\D)$ and $\TS(\G',\D,\K)=S_B(\G',\D)$ by definition.
    \\
    We have that the corresponding equivalence class $\E$ to the restrictive equivalence class $\EK$ contains $\G,\G'\in\E$. We have by the score equivalence of BIC (Proposition \ref{prop:BICallProps}) that $S_B(\G,\D)=S_B(\G',\D)$.
\end{proof}
\subsection{Proof of Lemma \ref{lem:TGESResult}} \label{supmat:subsubsec:StageI}
We start out with restating the Lemma we wish to prove:
\begin{customlem}{\ref{lem:TGESResult}}
   TGES results in a tiered MPDAG. 
\end{customlem}
First, we prove the following lemmas:
\begin{lemma} \label{lem:StageIAgree}
    Let $S$ be a $\K$-consistent, $\K$-score equivalent and decomposable scoring criterion. 
    Then stage (i) in TGES using $S$ applied to a restricted equivalence class $\EK$ will terminate in a CPDAG that is in agreement with the tiered background knowledge $\K$.
\end{lemma}
\begin{proof}
    Let $\EK$ be a restricted equivalence class that encodes $\K$, and $S$ a $\K$-consistent scoring criterion. We have that all DAGs $\G\in\EK$ encode the tiered background knowledge. Let $\E'$ be the equivalence class that results from stage (i) given $\EK$.
    \\
    Assume for contradiction that the equivalence class $\E'$ is not in agreement with $\K$.
    \\
    This implies that all DAGs $\G'\in\E'$ contradict $\K$. By $\K$-consistency ($(iii)$ in Def. \ref{def:K_Cons}) we have that $S(\G^*,\D,\K) > S(\G',\D,\K)$. But since GES is greedy and $S(\G^*,\D,\K) > S(\G',\D,\K)$ there is a contradiction, and hence it must be that $\E'$ is in agreement with $\K$.
\end{proof}
\begin{lemma} \label{lem:TGESStepResult}
    Each step in TGES results in a tiered MPDAG
\end{lemma}
\begin{proof}
    TGES commences with an empty graph which by default is a tiered MPDAG for any tiered background knowledge $\K$.
    \\
    Assuming we start with a tiered MPDAG, we prove that stage (i)-(iii) results in a tiered MPDAG. By Lemma \ref{lem:StageIAgree} Stage (i) results in a CPDAG in agreement with $\K$. Stage (ii) restricts the CPDAG according to $\K$ (Algorithm \ref{alg:restrict}) resulting in a PDAG. Stage (iii) utilizes Meek's rule 1, which by \citet{bang2023wiser} results in a tiered MPDAG.
\end{proof}
The proof of Lemma \ref{lem:TGESResult} now follows directly from Lemma \ref{lem:TGESStepResult}.

\subsection{Proof of Theorem \ref{Thm:TGES_ALL_S_C}: Sound and Completeness of TGES} \label{supmat:subsec:TGES_S_C}
We start out by restating the theorem we wish to prove:
\begin{customtheorem}{\ref{Thm:TGES_ALL_S_C}}
    (Sound and Completeness of TGES)\\
TGES using a $\K$-score equivalent, decomposable, $\K$-consistent and $\K$-locally consistent score criterion, results in a tiered MPDAG and for $n \rightarrow \infty$ the tiered MPDAG is almost surely a sound and complete estimate of the restricted Markov equivalence class of the true data-generating DAG.
\end{customtheorem}
This subsection will prove Theorem \ref{Thm:TGES_ALL_S_C} by introducing and proving three lemmas for each of the phases of TGES. Theorem \ref{Thm:TGES_ALL_S_C} follows directly from these lemmas. The first lemma concerns the forward phase of TGES:

\begin{lemma}\label{Lem:TGES_for_S_C}
(TGES forward phase)\\
Let $\GS$ be the DAG from which the data $\D$ was generated, $n$ be the number of observations in $\D$, and $\K$ the tiered background knowledge.
Let $\EK$ be the restricted equivalence class that is a result of a forward phase in TGES on $\D$ and $\K$ commencing with the empty graph. Then for $n \rightarrow \infty$ we have that all $\G\in \EK$ encode $\K$ and $\GS \leq \G$ almost surely. 
\end{lemma} 

The second lemma is concerned with the backward phase:

\begin{lemma} \label{Lem:TGES_back_S_C}
(TGES backward phase)\\
Let $\GS$ be the DAG from which the data $\D$ was generated, $n$ be the number of observations in $\D$, and $\K$ the tiered background knowledge.
Let $\EK$ be the restricted equivalence class that is a result of a backward phase in TGES on $\D$ and $\K$ commencing with a restricted equivalence class where all graphs $\G$ encode $\K$ and $\GS \leq \G$. Then for $n \rightarrow \infty$ we have that all $\HH\in \EK$ encode $\K$ and $\GS \approx \HH$ almost surely. 
\end{lemma}

As can be seen, TGES with only forward and backward phases fulfills Theorem \ref{Thm:TGES_ALL_S_C}. The third lemma then ensures that the turning phase does not turn any edges when the graph is Markov equivalent to the data-generating DAG.

\begin{lemma} \label{Lem:TGES_turn_S_C}
(TGES turning phase)\\
Let $\GS$ be the DAG from which the data $\D$ was generated, $n$ be the number of observations in $\D$, and $\K$ the tiered background knowledge.
Let $\EK$ be the restricted equivalence class that is a result of a turning phase in TGES on $\D$ and $\K$ commencing with a restricted equivalence class where all graphs $\HH$ encode $\K$ and $\GS \approx \HH$. Then for $n \rightarrow \infty$ we have that all $\mathcal{W}\in \EK$ encode $\K$ and $\GS \approx \mathcal{W}$ almost surely.
\end{lemma}

\subsubsection{Proof of Lemma \ref{Lem:TGES_for_S_C}: TGES forward phase}
In this subsection, we prove that as $n\rightarrow\infty$ the forward phase of TGES almost surely results in a tiered MPDAG which contains an independence statement only if the same statements is in the true data-generating DAG. The way we show this is by assuming the forward phase of TGES does not result in such a tiered MPDAG, which then leads to there existing an edge addition which improves the score. This is a contradiction to the forward phase continuing until no neighboring restricted equivalence classes have a higher score.
To show these contradictions we split into quite a lot of cases. To give an overview of the cases we illustrated them in Figure \ref{fig:Lemma_overview}. We also made minimal examples of pairs of a data-generating DAG and a DAG evaluated by the algorithm which correspond to each of the cases in Table \ref{tab:Lemma1_proof_examples}.
\\
\\
First, we introduce some needed concepts.
\\
\\
Note that for $\GS$ being the data-generating DAG, we can substitute $\perp\!\!\!\perp_p$ and $\not\!\perp\!\!\!\perp_p$ for $\indGS$ and $\nindGS$ respectively, as the distribution $p$ is assumed to be Markov and faithful w.r.t. $\GS$. We will use that these two are interchangeable in the proof of Lemma \ref{Lem:TGES_for_S_C}.
\\
We remind the reader that for any $X,Y$ in $\G$ where $Y\in\ndeG_X$ we have by the Markov property that $X\indG Y|\paG_X$. By contraposition of the Markov property, we thus have
\begin{align} \label{eq:negated_markov}
    X\nindG Y|\paG_X \Rightarrow Y\in \deG_X
\end{align}
We define a \textit{path} as in \citet{hernan2020causal} between two nodes $X$ and $Y$ in a DAG $\G$ as a sequence of edges that connects $X$ and $Y$ such that no node is visited more than once.
\\
A path $\pi$ is \textit{blocked} given a set of nodes $C$ if and only if it contains a non-collider that is in $C$, or it contains a collider that is not in $C$ and the collider has no descendants in $C$ \citep{Pearl_2009}. A path is \textit{open} if it is not blocked.
\\
Let $\pi$ be a path, we denote the length of $\pi$ as $|\pi|$, which is equal to the number of nodes visited including the endpoints.
\\
\\
We introduce two lemmas that will help prove Lemma \ref{Lem:TGES_for_S_C}. 
\begin{lemma} \label{lem:Open_subpath}
    Let $\pi$ be an open path in $\G$ given nodes $\mathbf{M}$. Then any sub-path $\pi_s$ of $\pi$ is open in $\G$ given $\mathbf{M}$ if the endpoints of $\pi_s$ are non-colliders on $\pi$.
\end{lemma}
\begin{proof}
    All non-colliders on $\pi$ are not contained in $\mathbf{M}$, otherwise $\pi$ would be blocked given $\mathbf{M}$. Let $A$ and $B$ be non-colliders on $\pi$ and $\pi_s$ be a sub-path of $\pi$ such that $A$ and $B$ are the end points of $\pi_s$. 
    \\
    All colliders on $\pi$ are opened by $\mathbf{M}$, hence all colliders on $\pi_s$ are opened by $\mathbf{M}$. Therefore, $\pi_s$ is not blocked given $\mathbf{M}$ and is hence open.
\end{proof}
\begin{lemma} \label{lem:nde_valid_dag}
    Let $\G$ be a DAG and $X$ a node in $\G$. Let $Y\in \ndeG_X$. Then the graph $\G'$ that results from adding the edge $\left\{Y\rightarrow X\right\}$ to $\G$ is also a DAG.
\end{lemma}
\begin{proof}
We are adding a directed edge, hence all edges in $\G'$ will have a direction. We need only to prove that adding the edge does not introduce a cycle. For there to be a cycle introduced by adding the edge $\left\{Y\rightarrow X\right\}$ in $\G$, there must be a directed path from $X$ to $Y$. But then we have that $Y \in\deG_X$ which is a contradiction, hence we do not introduce any cycles by adding the edge.
\end{proof}
Now we restate Lemma \ref{Lem:TGES_for_S_C} and then provide the proof. 
\begin{customlem}{\ref{Lem:TGES_for_S_C}}
(TGES forward phase)\\
Let $\GS$ be the DAG from which the data $\D$ was generated, $n$ be the number of observations in $\D$, and $\K$ the tiered background knowledge.
Let $\EK$ be the restricted equivalence class that is a result of a forward phase in TGES on $\D$ and $\K$ commencing with the empty graph. Then for $n \rightarrow \infty$ we have that all $\G\in \EK$ encode $\K$ and $\GS \leq \G$ almost surely.
\end{customlem} 
\begin{proof}
Each step in the TGES forward phase is ensured by Lemma \ref{lem:TGESStepResult} to result in a restricted equivalence class that encodes $\K$. 
\\
\\
We will make a proof by contradiction, where we assume that TGES forward phase results in a $\EK$ where $\exists\G\in \EK$ s.t. $\GS\not\leq \G$. We will prove that if this holds, then there exists a single edge addition to $\G$ that results in a DAG with a higher score. This contradicts forward TGES running until no further improvement of the score.
\\
\\
For $\G\in\EK$ assume $\GS\not\leq \G$. Then there exists some node $X$ and a node $Y\in\ndeG_X$ such that
\begin{align*}
    X \nindGS Y | \paG_X
\end{align*}
If $\tau(Y) \leq \tau(X)$, then adding the edge $\left\{Y \rightarrow X\right\}$ will by local $\K$-consistency ($(i)$ in Def. \ref{def:LocConsTier}) improve the score, which is a contradiction to TGES adding edges till no further increase in score. Since $Y\in\ndeG_X$ we are ensured by Lemma \ref{lem:nde_valid_dag} that adding the edge to $\G$ results in a DAG, which completes the proof for this case.
\\
\\
Instead, assume $\tau(X)<\tau(Y)$. We will split into four mutually exclusive cases of the d-separation statements of $X$ and $Y$ given $\paG_Y$ in $\G$ and $\GS$. We illustrate the cases by the Table \ref{tab:lemma_cases_ABCD} below:
\\
\begin{table}[H]
    \centering
    \begin{tabular}{c|cc}
     &$ Y \nindGS X | \paG_Y$ & $Y \indGS X | \paG_Y$ \\  
     \hline
    $Y \nindG X | \paG_Y$ & \textbf{Case A} & \textbf{Case B} \\
    $Y \indG X | \paG_Y$ & \textbf{Case C} & \textbf{Case D}
    \end{tabular}
    \caption{Cases for proof of Lemma \ref{Lem:TGES_for_S_C} when $\tau(X) < \tau(Y)$.}
    \label{tab:lemma_cases_ABCD}
\end{table}
\textbf{Case A and Case B}:
\\
In both cases, we have that $Y \nindG X | \paG_Y$ which by the contrapositioned Markov property (\ref{eq:negated_markov}) means that $X \in \deG_Y$. This implies that $\tau(X)\geq \tau(Y)$ which contradicts the assumption that $\tau(X) < \tau(Y)$. Hence, Case A and Case B can not occur.
\\
\\
\textbf{Case C}:
\\
In this case $Y \indG X | \paG_Y$ and $Y \nindGS X | \paG_Y$. We have by $\tau(X)<\tau(Y)$ that $X \in \ndeG_Y$, and by Lemma \ref{lem:nde_valid_dag} adding the edge $\left\{X\rightarrow Y\right\}$ to $\G$ will result in a DAG.
Since $Y \nindGS X | \paG_Y$, we have by local $\K$-consistency ($(i)$ in Def. \ref{def:LocConsTier}) that adding the edge $\left\{X\rightarrow Y\right\}$ will improve the score, which is a contradiction to TGES adding edges till no further increase in score.
\\
\\
\textbf{Case D}:
\\
First, let us summarize \textbf{Case D}:
\begin{align}
    Y \in \ndeG_X \label{as:Y_NDX}
    \\
    X \nindGS Y | \paG_X \label{as:X_nindGS_Y_paX}
    \\
    \tau(X)<\tau(Y) \label{as:tau_X_Y}
    \\
    Y \indG X | \paG_Y \label{as:Y_indG_X_paY}
    \\
    Y \indGS X | \paG_Y \label{as:Y_indGS_X_paY}
\end{align}
By $X \nindGS Y | \paG_X$ (\ref{as:X_nindGS_Y_paX}) there exist an open path $\pi$ in $\GS$ from $X$ to $Y$ given $\paG_X$. We write the path as a sequence of nodes, ordered from $X$ to $Y$, such that $\pi = \{X, ... , Y\}$.
\\
By $Y \indGS X | \paG_Y$ (\ref{as:Y_indGS_X_paY}) the path $\pi$ in $\GS$ is not open given $\paG_Y$. The path $\pi$ can be blocked by $\paG_Y$ in two non-exclusive ways:
\begin{enumerate}
    \item (Non-collider closed by $\paG_Y$) There exist a non-collider $V \in \pi$ where $V\in\paG_Y$.
    \item (Collider not opened by $\paG_Y$) There exist a collider $C \in \pi$ where for $\forall D \in \deGS_C \text{ we have that } D \notin \paG_Y$.
\end{enumerate}
Now we split into these two sub-cases. Starting from $Y$ and following the path $\pi$ backwards towards $X$, the \textbf{first} blockade is either:
\begin{enumerate}
    \item[]\textbf{Case D.1.} A non-collider $V \in \pi$ where $V\in\paG_Y$.
    \item[]\textbf{Case D.2.} A collider $C \in \pi$ where for $\forall D \in \deGS_C \text{ we have that } D \notin \paG_Y$. 
\end{enumerate}
These two cases are mutually exclusive since we are only looking at the first blockade on $\pi$ starting from $Y$ and going towards $X$. The two cases cover all scenarios since there are no other way a path can be blocked.
\\
\\
\textbf{Case D.1.}
\\
We assume the first blockade on the path $\pi$ starting from $Y$ and following the path backward towards $X$ is a non-collider $V \in \pi$ where $V\in\paG_Y$. Then the following hold:
\begin{align}
    V\in \ndeG_X \label{D1:V_ndeX}
    \\
    X \nindGS V |\paG_X \label{D1:X_nindGS_V_pa_X}
    \\
    X \notin \paG_V  \label{D1:X_n_paV}
\end{align}
We prove each statement:
\\
\\
(\ref{D1:V_ndeX}): We see that $V\in \ndeG_X$ (\ref{D1:V_ndeX}) since $V\in\paG_Y$ and $\paG_Y \subseteq \ndeG_X$. If $\paG_Y \not\subseteq \ndeG_X$ then there exist a $Z\in\paG_Y$ such that $Z\in \deG_X$, but then $Y\in\deG_X$ which contradicts $Y\in\ndeG_X$ (\ref{as:Y_NDX}).
\\
\\
(\ref{D1:X_nindGS_V_pa_X}): $X$ and $V$ are both non-colliders on $\pi$, hence we have by Lemma \ref{lem:Open_subpath} that the sub-path $\pi_V = \{X,...,V\}$ of $\pi$ is an open path in $\GS$ given $\paG_X$. 
Since there is an open path $\pi_V$ in $\GS$ given $\paG_X$ 
we have that $X \nindGS V |\paG_X$ (\ref{D1:X_nindGS_V_pa_X}).
\\
\\
(\ref{D1:X_n_paV}): $X \notin \paG_V$ (\ref{D1:X_n_paV}) by the following reasoning. If it were that $X \in \paG_V$, and we have that $V \in \paG_Y$, then it would be that $Y \in \deG_X$, which contradicts $Y \in \ndeG_X$ (\ref{as:Y_NDX}).
\\
\\
Now we consider the two distinct sub-cases:
    \begin{enumerate}
        \item[]\textbf{Case D.1.1.} $\tau(V)\leq \tau(X)$
        \item[]\textbf{Case D.1.2.} $\tau(V) > \tau(X)$
    \end{enumerate}
\textbf{Case D.1.1.}
\\
As $V \in \ndeG_X$ (\ref{D1:V_ndeX}), we get from Lemma \ref{lem:nde_valid_dag} that adding the edge $\left\{V \rightarrow X\right\}$ will result in a DAG.
Since $\tau(V)\leq \tau(X)$  and $X\nindGS V|\paG_X$ (\ref{D1:X_nindGS_V_pa_X}), we have by local $\K$-consistency ($(i)$ in Def. \ref{def:LocConsTier}) that adding the edge $\left\{V \rightarrow X\right\}$ will improve the score, which is a contradiction to TGES adding edges till no further increase in score.
\\
\\
\textbf{Case D.1.2.}
\\
First we see that since $\tau(V) > \tau(X)$, it must be that
\begin{align}
    X \in \ndeG_V \label{D12:X_ndeV}
\end{align}
\\
We split in two further sub-cases: % $\text{sub}^3$-cases 
\begin{enumerate}
    \item[]\textbf{Case D.1.2.1.} $V \nindGS X | \paG_V$
    \item[]\textbf{Case D.1.2.2.} $V \indGS X | \paG_V$
\end{enumerate}
\textbf{Case D.1.2.1.}
\\
By $X \in \ndeG_V$ (\ref{D12:X_ndeV}) and Lemma $\ref{lem:nde_valid_dag}$ adding the edge $\left\{X \rightarrow V\right\}$ to $\G$ results in a DAG. Since $\tau(V) > \tau(X)$ and $V \nindGS X | \paG_V$, we have by local $\K$-consistency ($(i)$ in Def. \ref{def:LocConsTier}) that adding the edge $\left\{X \rightarrow V\right\}$ will improve the score, which is a contradiction to TGES adding edges till no further increase in score.
\\
\\
\textbf{Case D.1.2.2.}
\\
Now note that we are in the exact same scenario as at the start of \textbf{Case D}. 
Therefore, set $Y:=V$ and fix the path $\pi := \pi_V$ and go to the start of \textbf{Case D}. 
\\
\\
Now we show that if we continue setting $Y:=V$, we at some point end up in another case than \textbf{Case D.1.2.2.}. 
\\ 
Since $V\in \paG_Y$ we know that $V\neq Y$, hence $|\pi_V|<|\pi|$.
\\
Since $\pi$ is finite and $|\pi_V|<|\pi|$, we just need to show that for a $|\pi_V|$ smaller than a specific integer, we can must end up in another case than \textbf{Case D.1.2.2.}. We show this for $|\pi_V|<3$. 
\\
Assume we are in \textbf{Case D.1.2.} and $|\pi_V|<3$. A path must be at least of length $2$, hence $|\pi_V|=2$. Thus, the path in $\GS$ consists of only one edge between $X$ and $V$, and since we assumed $\tau(X)<\tau(Y)$ in this case, it must be the edge $\left\{X \rightarrow V\right\}$ in $\GS$. And since $X \notin \paG_V$ (\ref{D1:X_n_paV}) we have that $V \nindGS X|\paG_V$, which mean we are in \textbf{Case 1.2.1.}. 
\\
\\
\textbf{Case D.2.}
\\
This is the case where starting from $Y$ and following the path $\pi$ backward towards $X$, the first blockade is a collider $C \in \pi$ where for $\forall D \in \deGS_C \text{ we have that } D \notin \paG_Y$. In particular, we know that $C\notin \paG_Y$, since $C\in \deGS_C$.
\\\\
Since $\pi$ is an open path given $\paG_X$ and $C$ is a collider on the path, there must exist some $D\in \deGS_C$ such that $D\in\paG_X$, otherwise $\pi$ is blocked given $\paG_X$.
\\
Since $D\in \deGS_C$ we know $\tau(C)\leq\tau(D)$, and since $D\in\paG_X$ we know $\tau(D)\leq\tau(X)$. Lastly, we know that $\tau(X)<\tau(Y)$ (\ref{as:tau_X_Y}), and hence it follows that $\tau(C)<\tau(Y)$. 
\\
\\
Now let $K$ be the node to the right in the collider structure of $C$, such that $\rightarrow C\leftarrow K$ and $\pi = \{X,...,C,K,...,Y\}$. 
Then $K \in \paGS_C$, hence $\tau(K)\leq \tau(C)$, and in particular, $\tau(K)<\tau(Y)$. Since $\tau(K)<\tau(Y)$ we have that $K\in \ndeG_Y$. 
\\
\\
It must be that $K$ is a non-collider on the path $\pi$, since it is next to a collider $C$. Starting from $Y$, $K$ came before $C$, and because we are not in \textbf{Case D.1.}, we know that $\K\notin \paG_Y$. 
\\
Since $C$ was the first blockade given $\paG_Y$ on $\pi$ starting from $Y$ in $\GS$, we know that $\pi_K=\{K,...,Y\}$ is an open path in $\GS$ given $\paG_Y$. Since there is an open path $\pi_K$ in $\GS$ given $\paG_Y$ we have that $Y \nindGS K|\paG_Y$.
\\
Since $K\in \ndeG_Y$, we have by Lemma \ref{lem:nde_valid_dag} that adding the edge $\left\{K \rightarrow Y\right\}$ results in a DAG.
\\
Since $\tau(K)<\tau(Y)$ and $Y \nindGS K|\paG_Y$ we have by local $\K$-consistency ($(i)$ in Def. \ref{def:LocConsTier}) that adding the edge $\left\{K \rightarrow Y\right\}$ will improve the score, which is a contradiction to TGES adding edges till no further increase in score.
\end{proof}

\subsubsection{Proof of Lemma \ref{Lem:TGES_back_S_C}: TGES backward phase} 
This subsection introduces everything needed to prove Lemma \ref{Lem:TGES_back_S_C} and states the proof of it. 
\\\\
First, we introduce some concepts and notations from \citet{Chickering2002}. 
We will say that an edge $\left\{X \rightarrow Y\right\}$ in $\G$ is \textit{covered} if $\text{Pa}^\G_Y = \text{Pa}^\G_X \cup X$.
\\ Covered edges have the following property as described in \citet{Chickering2002}:
Let $\G$ be any DAG, and let $\G'$ be the result of reversing the edge $\left\{X \rightarrow Y\right\}$ in $\G$. Then $\G'$ is a DAG that is equivalent to $\G$ if and only if $\left\{X \rightarrow Y\right\}$ is covered in $\G$.
\\
For any subset $\mathbf{A}$ of the nodes 
in $\G$, we say that a node $A \in \mathbf{A}$ is \textit{maximal} if there is no other node 
$A' \in \mathbf{A}$ such that $A'$ is an ancestor of $A$ in $\G$.
\\
We will denote a node as a \textit{sink node} if all nodes adjacent to it are parents of the node.
\\
\\
Next, we introduce the following theorem from \citet{Chickering2002} which we will extend for tiered background knowledge.
\begin{theorem} \label{thm4chick}
    (Theorem 4 in \citep{Chickering2002})
    \\
    Let $\G$ and $\HH$ be any pair of DAGs such that $\G\leq\HH$. Let $r$ be the number of edges in $\HH$ that have opposite orientations in $\G$, and let $m$ be the number of edges in $\HH$ that do not exist in either orientation in $\G$. Then, there exists a sequence of at most $r+2m$ edge reversals and additions in $\G$ with the following properties:
\begin{enumerate}
    \item Each edge reversed is a covered edge 
    \item After each reversal and addition $\G$ is a DAG and $\G\leq\HH$ 
    \item After all reversals and additions $\G = \HH$
\end{enumerate}
\end{theorem}
\citet{Chickering2002} shows that such a sequence of reversals and additions exist by showing each step in the sequence can be produced by Algorithm \ref{alg:apply} with two DAGs $\G$ and $\HH$ where $\G\leq\HH$. He shows that Algorithm \ref{alg:apply} will result in a DAG $\G'$ where: $\G'\leq\HH$ and $\G'$ is "closer" to $\HH$ than $\G$, in the sense that either the number of edges of opposite orientation has decreased by 1 and the number of different adjacencies remained the same or the number of different adjacencies has decreased by 1 and the number of edges of opposite orientation has maximally increased by one. 
\\
\\
We restate the algorithm from \citet{Chickering2002} and include "$\triangleright\#$" for references used in a later proof.
\begin{algorithm}[H]
\caption{Apply-Edge-Operation \citep{Chickering2002}}
\label{alg:apply}
\begin{algorithmic}  
\Statex \textbf{Input:} DAGs $\mathcal{G}, \HH$
\State $\mathcal{G}'\gets \mathcal{G}$
\While{ $\mathcal{G}$ and $\HH$ contain a node $Y$ that is a sink in both DAGs and for which $\text{Pa}^{\mathcal{G}}_Y = \text{Pa}^{\HH}_Y$}
\State Remove $Y$ and all edges with an endpoint at $Y$ from both DAGs. \Comment{1.}

\EndWhile
\State Let $Y$ be any sink node in $\HH$

\If{$Y$ has no children in $\mathcal{G}$}
\State Let $X$ be any parent of $Y$ in $\HH$ that is not a parent of $Y$ in $\mathcal{G}$. 
\State Add the edge $\left\{X \rightarrow Y\right\}$ to $\mathcal{G}'$ \Comment{2.}
\Else
\State Let $D \in \text{De}^{\mathcal{G}}_Y$ denote the (unique) maximal element from this set within $\HH$. 
\State Let $Z$ be any maximal child of $Y$ in $\mathcal{G}$ such that $D$ is a descendant of $Z$ in $\mathcal{G}$.
\If{$\left\{Y \rightarrow Z\right\}$ is covered in $\mathcal{G}$}
\State Reverse $\left\{Y \rightarrow Z\right\}$ in $\mathcal{G}'$ \Comment{3.}

\ElsIf{there exists a node $X$ that is a parent of $Y$ but not a parent of $Z$ in $\mathcal{G}$}
\State Add $\left\{X \rightarrow Z\right\}$ to $\mathcal{G}'$ \Comment{4.}
\Else
\State Let $X$ be any parent of $Z$ that is not a parent of $Y$ in $\G$. 
\State Add $\left\{X \rightarrow Y\right\}$ to $\mathcal{G}'$. \Comment{5.}
\EndIf
\EndIf
\State \textbf{Return} $\G'$
\end{algorithmic}
\end{algorithm}
We introduce a Proposition also from \citet{Chickering2002}.
\begin{proposition} \label{prop:chick27} \citep{Chickering2002}
Let $\G$ and $\HH$ be two DAGs such that $\G \leq \HH$. If there is an edge between $X$ and $Y$ in $\G$, then there is an edge between $X$ and $Y$ in $\HH$.    
\end{proposition}
We show in the following theorem, that if $\G$ and $\HH$ encode $\K$ then there exists a sequence of edge addition and reversals where each operation results in a DAG which encodes $\K$. The theorem is an extension of Theorem \ref{thm4chick} and we will prove it using the same algorithm (Algorithm \ref{alg:apply}) to construct the sequence.
\begin{theorem} \label{TIERthm4chick}
    Let $\G$ and $\HH$ be any pair of DAGs such that $\G$ $\leq$ $\HH$\textbf{ and they both encode the tiered background knowledge $\K$}. Let r be the number of edges in $\HH$ that have opposite orientation in $\G$, and let m be the number of edges in $\HH$ that do not exist in either orientation in $\G$. Then, there exists a sequence of at most $r+2m$ edge reversals and additions in $\G$ with the following properties:
\begin{enumerate}
    \item Each edge reversed is a covered edge 
    \item After each reversal and addition $\G$ is a DAG \textbf{that encodes} $\K$ and $\G $ $\leq$ $ \HH$ 
    \item After all reversals and additions $\G = \HH$
\end{enumerate}
\end{theorem} 
\begin{proof}
    We wish to show that by running Algorithm \ref{alg:apply} at most $r+2m$ times, we obtain a sequence with properties 1, 2 and 3. By Theorem \ref{thm4chick} we have that 1 and 3 are fulfilled, and that after each iteration of running Algorithm \ref{alg:apply} $\G\leq\HH$. We also have that $\G=\HH$ after, at most, $r+2m$ times. We thus only need to show that $\G$ encodes $\K$ after each reversal or addition performed by Algorithm \ref{alg:apply}.
    \\
    We remind the reader of the relationship between $\G$ and $\HH$, as this will be used in the proofs. All d-separation statements in $\HH$ also hold in $\G$, i.e. $A\perp\!\!\!\perp_\G B|C \Rightarrow A\perp\!\!\!\perp_\HH B|C$. Furthermore, we have that both DAGs encode $\K$, and hence the nodes have the same tiered ordering $\tau$ in both DAGs.
    \\
    \\
    The proof is divided into the 5 cases where a modification is made to the DAG $\G$ by Algorithm \ref{alg:apply}, and we need to show that this modification does not contradict $\K$. The 5 cases are denoted with a numbered comment, $\triangleright$\#, in the right-hand side of Algorithm \ref{alg:apply}.
    \\
    \\
    \Comment{1.} First, we realize that by removing a sink node and all incident edges from a DAG that encodes $\K$ we obtain a DAG that still encodes $\K$ as no contradiction to tiered background knowledge can occur by removing edges.
    \\\\
    \Comment{2.}  Here $X\in\text{Pa}^{\HH}_Y$ and it must be that $\tau(X)\leq \tau(Y)$ hence the edge $\left\{X \rightarrow Y\right\}$ encodes $\K$.
    \\\\
    \Comment{3.} Here $Y$ is a sink node in $\HH$, $D$ is maximal in $\HH$ in the set $\text{De}^{\mathcal{G}}_Y$, $D\in\text{De}^{\mathcal{G}}_Y$, $D\in\text{De}^{\mathcal{G}}_Z$ and $Z\in\text{Ch}^{\mathcal{G}}_Y$. 
    \\
    \\
    We have that $Z\neq Y$ by definition. We also have that $Y \neq D$. We see this by the following reasoning: We know that $Y$ has children in $\G$, otherwise, we would have terminated in case $\triangleright$ 1. Let $A\in\text{Ch}^{\mathcal{G}}_Y$. Then $A$ is trivially also a descendant of $Y$ in $\G$. By Proposition \ref{prop:chick27} there is an edge between $Y$ and $A$ in $\HH$, which since $Y$ is a sink node, must go into $Y$. Therefore, $A\in \text{An}_Y^\HH$ which means $Y$ is not maximal in $\HH$ in the set $\text{De}^{\mathcal{G}}_Y$. $D$ is chosen to be maximal which means $Y \neq D$.
    \\
    \\
    By $Z\in\text{Ch}^{\mathcal{G}}_Y$ we have that $\tau(Y)\leq \tau(Z)$. We wish to show that $\tau(Y) = \tau(Z)$ as this would mean that reversing the edge is not in contradiction with $\K$.
    \\
    We first assume that $\tau(Y) < \tau(Z)$. 
    \\
    \\
    Then, since $Y$ is a sink node in $\HH$, all nodes are non descendants of $Y$ in $\HH$, and in particular $D\in\text{NDe}^{\HH}_Y$. Therefore $Y\perp\!\!\!\perp_\HH D|\text{Pa}^{\HH}_Y$ holds and since $\G \leq \HH$ we have that $Y\perp\!\!\!\perp_\G D|\text{Pa}^{\HH}_Y$. We have that $D\notin \text{Pa}^\HH_Y$ as $D\in\text{De}^{\mathcal{G}}_Z$, which gives $\tau(Z)\leq\tau(D)$, and by assumption $\tau(Y)<\tau(Z)$, and an edge from $D$ to $Y$ would contradict $\K$.
    \\
    Since $Z\in\text{Ch}^{\mathcal{G}}_Y$ and $D\in\text{De}^{\mathcal{G}}_Z$ there must be a directed path in $\G$ from $Y$ to $D$, starting with a directed edge from $Y$ to $Z$.
    \\
    As $Y$ and $D$ are d-separated in $\G$ given $\text{Pa}^{\HH}_Y$, there exist a node $V \in \text{Pa}^{\HH}_Y$ on the directed path from $Y$ to $D$ through $Z$ in $\G$. $\tau(Z) \leq \tau(V)$ since $V$ is on the directed path in $\G$ which start with $Z$, but since $V\in \text{Pa}^{\HH}_Y$ we also have that $\tau(V)\leq \tau(Y)$. 
    \\
    We then have that $\tau(V)\leq \tau(Y)<\tau(Z)\leq \tau(V)$, but $\tau(V)<\tau(V)$ is a contradiction to a node only belonging to a single tier, and it hence must be that $\tau(Z) = \tau(Y)$, which implies that the edge $\left\{Y\rightarrow Z\right\}$ can be reversed while still encoding $\K$.
    \\
    \\
    \Comment{4.} Here $X\in\text{Pa}^{\mathcal{G}}_Y$ and $Z\in\text{Ch}^{\mathcal{G}}_Y$ it must be that $\tau(X)\leq \tau(Y) \leq \tau(Z)$. Hence the added edge $\left\{X \rightarrow Z\right\}$ encodes $\K$.
    \\\\
    \Comment{5.} Here the setup is the same as described in $\triangleright$ 2 in addition to $X\in\text{Pa}^{\mathcal{G}}_Z$.
    We wish to show that $\tau(X) \leq \tau(Y)$, and we assume for contradiction that $\tau(Y) < \tau(X)$. 
    \\
    Since $X\in\text{Pa}^{\mathcal{G}}_Z$, we have that $\tau(X)\leq \tau(Z)$, hence $\tau(Y)<\tau(Z)$ and we can then use same procedure as in $\triangleright$ 2 to obtain a contradiction. 
    \\
    It therefore must be that $\tau(X)\leq \tau(Y)$, and thus the edge $\left\{X \rightarrow Y\right\}$ encodes $\K$.
    \end{proof}
We include one more lemma needed for the proof of Lemma \ref{Lem:TGES_back_S_C}.
\begin{lemma}\label{lem:needed_for_back}
    Let $\G$ be a DAG and let $\G'$ be the DAG that is a result of either a covered edge reversal or an edge addition in $\G$. Then $\G\leq \G'$.
\end{lemma}
\begin{proof}
    Assume the modification made was a covered edge reversal. Then $\G$ and $\G'$ are Markov equivalent and in particular $\G\leq \G'$.
    \\
    Instead, assume the modification was an edge addition. For $\G\not\leq \G'$, there must exist an open path in $\G$ which is blocked in $\G'$. Adding an edge only allows for there to exist a new blocked path, not to block an existing open path. Therefore $\G\leq \G'$.
\end{proof}
We restate Lemma \ref{Lem:TGES_back_S_C} and present a proof of it:
\begin{customlem}{\ref{Lem:TGES_back_S_C}}
(TGES backward phase)\\
Let $\GS$ be the DAG from which the data $\D$ was generated, $n$ be the number of observations in $\D$, and $\K$ the tiered background knowledge.
Let $\EK$ be the restricted equivalence class that is a result of a backward phase in TGES on $\D$ and $\K$ commencing with a restricted equivalence class where all graphs $\G$ encode $\K$ and $\GS \leq \G$. Then for $n \rightarrow \infty$ we have that all $\HH\in \EK$ encode $\K$ and $\GS \approx \HH$ almost surely. 
\end{customlem} 
\begin{proof}

    The TGES backward phase consists of steps provided in Algorithm \ref{alg:T_bac_ph} (illustrated in Figure \ref{fig:TGES}) where each step removes an edge from the tiered MPDAG. The phase continues to remove edges until there are no neighboring restricted equivalence classes with a greater score.
    \\
    First, we see that each step results in a restricted equivalence class which encodes $\K$ by Lemma \ref{lem:TGESStepResult}.
    \\
    Next, see that after each step in the phase, all DAGs $\HH$ in the restricted equivalence class have that $\GS\leq\HH$.
    Assume that this is not the case, which means there exists a DAG $\HH'$ in the restricted equivalence class such that $\GS \not\leq \HH'$. Then we have by $\K$-consistency ($(i)$ in Def. \ref{def:K_Cons}) that the score has decreased, which is a contradiction to the algorithm being greedy. 
    \\
    \\
    Let $\EK$ be the result of the backward phase in TGES. We need only to show that for all DAGs $\G \in \EK$ we have that $\G \leq \GS$.
    We will prove this by contradiction: Assume that $\exists\G'\in\EK$ s.t. $\G'\not\leq \GS$. 
    \\
    Let $\EKS$ be the restricted equivalence class which contains the true data-generating DAG $\GS$. Then $\G^*\in\EKS$ and $\G'\in\EK$. By $\G'\not\leq \GS$ there exists an independence statement in $\GS$ which is not in $\G'$, hence $\EKS \neq \EK$. 
    \\
    Since $\G^* \leq \G'$ and they both encode $\K$ we have by Theorem \ref{TIERthm4chick} that there exists a sequence of modifications that can turn $\G^*$ into $\G'$, such that after each modification of $\GS$, the modified DAG $\Tilde{\GS}$ will encode $\K$ and $\Tilde{\GS}\leq \G'$.
    Because of the relation $\EKS\neq \EK$, there must at least be one edge addition in the sequence, if there were only edge reversals, the restricted equivalence classes would be identical.
    Let $\Tilde{\GS}$ be the graph that precedes the very last edge addition in the sequence. Then $\Tilde{\GS}$ is contained in a neighboring restricted equivalence class (by definition) and we have by Lemma \ref{lem:needed_for_back} that $\GS\leq\Tilde{\GS}$. By $\K$-consistency ($(ii)$ in Def. \ref{def:K_Cons}) of the scoring criterion we have that $S(\Tilde{\GS}, \D,\K) > S(\G', \D,\K)$, which contradicts that the backward phase iterates until there is no larger score among the neighboring restricted equivalence classes. Hence the resulting restricted equivalence class must be $\EKS$.
\end{proof}
\subsubsection{Proof of Lemma \ref{Lem:TGES_turn_S_C}: TGES turning phase}
This subsection introduces everything needed to prove Lemma \ref{Lem:TGES_turn_S_C} and states the proof of it. 
\\
\\
We start by introducing a proposition, which uses Theorem \ref{TIERthm4chick}.
\begin{proposition}
    \label{prop:PerfMapOptiTBIC} 
    Let $\EKS$ be the restricted equivalence class that contains the true data-generating DAG from which the data $\D$ was generated, $n$ be the number of observations in $\D$, and $\K$ tiered background knowledge. 
    Then for a $\K$-score equivalent and $\K$-consistent scoring criterion $S$ we have that for $n\rightarrow \infty$,
    \begin{align}
        S(\EKS,\D)>S(\EK,\D)\text{ for all }\EKS \neq \EK \text{ almost surely.}
    \end{align}
\end{proposition}
\begin{proof}
     If $\EK$ does not encode $\K$, we have by $\K$-consistency ($(iii)$ in Def. \ref{def:K_Cons}) that $S(\EKS,\D) > S(\EK,\D)$.
     \\
     Now assume $\EK$ encodes $\K$. Assume for contradiction that there exist a $\EK$ such that $S(\EKS,\D) \leq S(\EK,\D)$ where $\EKS \neq \EK$. By $\K$-consistency ($(i)$ in Def. \ref{def:K_Cons}) we have that for $\G \in \EK$ and $\G^*\in \EKS$ that $\G^*\leq \G$. By Theorem \ref{TIERthm4chick} there exists a sequence of edge reversals and additions that transform $\G^*$ into $\G$. Since all edge reversals are ensured to result in a graph in the same restricted equivalence class, we have by $\K$-score equivalence of $S$ that the score remains the same for each edge reversal. After each edge addition the score increases following Lemma \ref{lem:needed_for_back} and by $\K$-consistency ($(ii)$ in Def. \ref{def:K_Cons}). Hence there can be no edge additions. This contradicts $\EKS \neq \EK$.
\end{proof}
We restate the lemma and present the proof
\begin{customlem}{\ref{Lem:TGES_turn_S_C}}
(TGES turning phase)\\
Let $\GS$ be the DAG from which the data $\D$ was generated, $n$ be the number of observations in $\D$, and $\K$ the tiered background knowledge.
Let $\EK$ be the restricted equivalence class that is a result of a turning phase in TGES on $\D$ and $\K$ commencing with a restricted equivalence class where all graphs $\HH$ encode $\K$ and $\GS \approx \HH$. Then for $n \rightarrow \infty$ we have that all $\mathcal{W}\in \EK$ encode $\K$ and $\GS \approx \mathcal{W}$ almost surely.
\end{customlem}
\begin{proof}
    By Proposition \ref{prop:PerfMapOptiTBIC} any change to the restricted equivalence class will result in a lower score, hence no changes are made by the turning phase.
\end{proof}
\newpage
\subsubsection{Figures}
\label{subsubsec:fig_for_Proof}
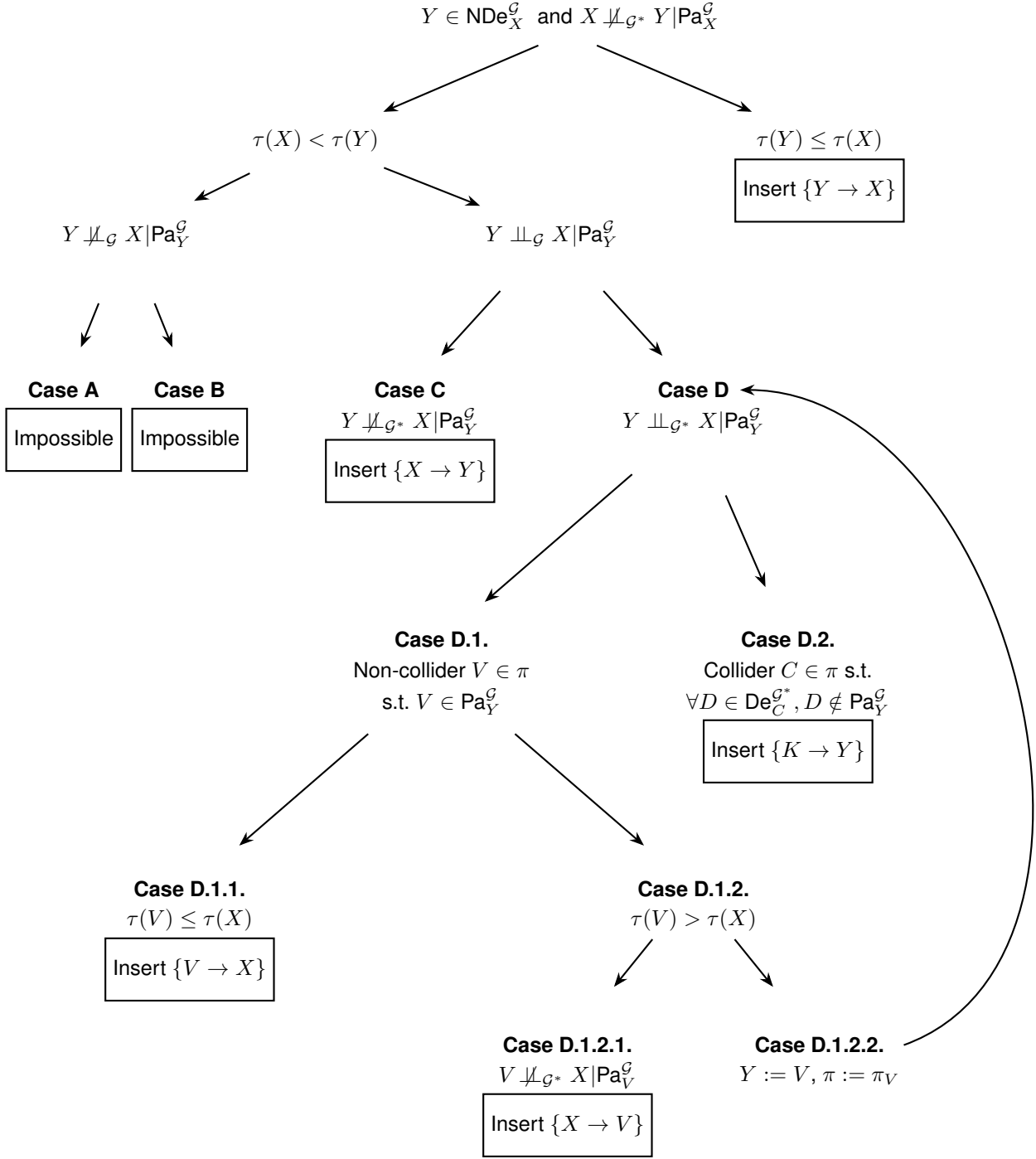
\begin{figure}[H]
\centering
\begin{tikzpicture}[>=Stealth, node distance=2cm, every node/.style={circle, draw, thick}]

  \tikzset{reg/.style={draw = none,font=\sffamily}}
    \tikzset{sqr/.style={
  draw, % Enables drawing the border. By default, the border is rectangular.
  font=\sffamily, % Keeps the sans-serif font
  rectangle, % Explicitly defines the shape as a rectangle
  minimum size=1cm, % Ensures that the node is at least a square of 1cm by 1cm
  align=center % Aligns the text in the center of the node
}}

%Nodes
  \node[reg] (Start01) at (-1.5,0) {$Y\in \ndeG_X$};
  \node[reg] (Start01) at (-0.2,0) {and};
  \node[reg] (Start01) at (1.3,0) {$X\nindGS Y|\paG_X$};
  \node[reg] (Start1) at (0.3,-0.4) {};
  \node[reg] (Start2) at (-0.3,-0.4) {};
  \node[reg] (YX) at (4,-2) {$\tau(Y)\leq\tau(X)$};
  \node[sqr] (Label1) at (4,-2.8) {Insert $\{Y\rightarrow X\}$};
  \node[reg] (XY) at (-4,-2) {$\tau(X) < \tau(Y)$};
  \node[reg] (Label11) at (-7,-3.5) {$Y\nindG X|\paG_Y$};
  
  \node[reg] (Label12) at (-0.25,-3.5) {$Y\indG X|\paG_Y$};
  
  \node[reg] (A) at (-8,-6) {\textbf{Case A}};
  \node[sqr] (Label2) at (-8,-6.8) {Impossible};
  \node[reg] (B) at (-6,-6) {\textbf{Case B}};
  \node[sqr] (Label3) at (-6,-6.8) {Impossible};
  \node[reg] (C) at (-2.5,-6) {\textbf{Case C}};
  \node[reg] (Label4) at (-2.5,-6.5) {$Y\nindGS X|\paG_Y$};
  \node[sqr] (Label5) at (-2.5,-7.3) {Insert $\{X\rightarrow Y\}$};
  \node[reg] (D) at (2,-6) {\textbf{Case D}};
  \node[reg] (Label6) at (2,-6.5) {$Y\indGS X|\paG_Y$};
  \node[reg] (D1) at (-2,-10) {\textbf{Case D.1.}};
  \node[reg] (Label71) at (-2,-10.5) {Non-collider $V\in \pi$};
  \node[reg] (Label72) at (-2,-11) {s.t. $V\in \paG_Y$};

  \node[reg] (D2) at (3.5,-10) {\textbf{Case D.2.}};
  \node[reg] (Label81) at (3.5,-10.5) {Collider $C\in \pi$ s.t.};
  \node[reg] (Label82) at (3.5,-11) {$\forall D\in\deGS_C, D\notin \paG_Y$ };
  \node[sqr] (Label83) at (3.5,-11.8) {Insert $\{K\rightarrow Y\}$};
  
  \node[reg] (D11) at (-6,-14) {\textbf{Case D.1.1.}};
  \node[reg] (Label91) at (-6,-14.5) {$\tau(V) \leq \tau(X)$};
  \node[sqr] (Label92) at (-6,-15.3) {Insert $\{V\rightarrow X\}$};
  \node[reg] (D12) at (2,-14) {\textbf{Case D.1.2.}};
  \node[reg] (Label10) at (2,-14.5) {$\tau(V) > \tau(X)$};
  \node[reg] (D121) at (0,-16.5) {\textbf{Case D.1.2.1.}};
  \node[reg] (Label111) at (0,-17) {$V\nindGS X|\paG_V$};
  \node[sqr] (Label112) at (0,-17.8) {Insert $\{X\rightarrow V\}$};
  
  \node[reg] (D122) at (4,-16.5) {\textbf{Case D.1.2.2.}};
  \node[reg] (Label121) at (4,-17) {$Y:=V$, $\pi:=\pi_V$};

%Edges
  \draw[->, thick] (Start1) -- (YX);
  \draw[->, thick] (Start2) -- (XY);
  \draw[->, thick] (XY) -- (Label11);
  \draw[->, thick] (XY) -- (Label12);
  \draw[->, thick] (Label11) -- (A);
  \draw[->, thick] (Label11) -- (B);
  \draw[->, thick] (Label12) -- (C);
  \draw[->, thick] (Label12) -- (D);
  \draw[->, thick] (Label6) -- (D1);
  \draw[->, thick] (Label6) -- (D2);
  \draw[->, thick] (Label71) -- (D11);
  \draw[->, thick] (Label71) -- (D12);
  \draw[->, thick] (D12) -- (D121);
  \draw[->, thick] (D12) -- (D122);
  \draw[->, thick] (Label121) to [out=20,in=0] (D);

\end{tikzpicture}
\caption{Illustration of Lemma \ref{Lem:TGES_for_S_C} proof. Here $\pi$ is chosen as a path from $X$ to $Y$ in $\GS$ that is open given $\paG_X$. This path exists by $X\nindGS Y|\paG_X$. $K$ is the node in the collider structure of $C$ closest to $Y$ on $\pi$. $\pi_V$ is the sub-path of $\pi$ which follows $\pi$ from $X$ to $V$.
} \label{fig:Lemma_overview}
\end{figure}
\begin{table}[H]
    \centering
    \begin{tabular}{|c|c|c|c|}
    \hline
    \textbf{Case} & $\G$ & $\GS$ & \textbf{Edge addition} \\
    \hline
    $\tau(Y)\leq\tau(X)$ & 
    \parbox[c][2cm][c]{4cm}{ % Adjust the size as needed
    \centering
    \begin{tikzpicture}[>=Stealth, node distance=2cm, every node/.style={circle, draw, thick}]

  \tikzset{reg/.style={draw = none,font=\sffamily}}

%Nodes
  \node[reg] (X) at (2,0) {$X$};
  \node[reg] (Y) at (0,0) {$Y$};
  \node[reg] (top1) at (1,0.5) {};
  \node[reg] (bot1) at (1,-0.5) {};
  \node[reg] (label1) at (0,-1) {$\tau = 1$};
  \node[reg] (label2) at (2,-1) {$\tau = 2$};
%Edges
%Tiers
  \draw[-,dotted]{} (top1) -- (bot1);

\end{tikzpicture}
    } &
    \parbox[c][2cm][c]{4cm}{ % Adjust the size as needed
    \centering
    \begin{tikzpicture}[>=Stealth, node distance=2cm, every node/.style={circle, draw, thick}]

  \tikzset{reg/.style={draw = none,font=\sffamily}}

%Nodes
  \node[reg] (X) at (2,0) {$X$};
  \node[reg] (Y) at (0,0) {$Y$};
  \node[reg] (top1) at (1,0.5) {};
  \node[reg] (bot1) at (1,-0.5) {};
  \node[reg] (label1) at (0,-1) {$\tau = 1$};
  \node[reg] (label2) at (2,-1) {$\tau = 2$};
%Edges
   \draw[->, thick] (Y) -- (X);
%Tiers
  \draw[-,dotted]{} (top1) -- (bot1);

\end{tikzpicture}
    } &
    \parbox[c][2cm][c]{2.5cm}{ % Adjust the size as needed
    \centering
    Insert $\left\{Y \rightarrow X\right\}$
    } \\
    \hline
    \textbf{Case C} & 
    \parbox[c][3.65cm][c]{4cm}{ % Adjust the size as needed
    \centering
    \begin{tikzpicture}[>=Stealth, node distance=2cm, every node/.style={circle, draw, thick}]

    \tikzset{reg/.style={draw = none,font=\sffamily}}

    %Nodes
    \node[reg] (X) at (0,1) {$X$};
    \node[reg] (Z) at (0,-1) {$Z$};
    \node[reg] (Y) at (2,0) {$Y$};
    \node[reg] (top1) at (1,1.5) {};
    \node[reg] (bot1) at (1,-1.5) {};
    \node[reg] (label1) at (0,-2) {$\tau = 1$};
    \node[reg] (label2) at (2,-2) {$\tau = 2$};
    %Edges
    %Tiers
    \draw[-,dotted]{} (top1) -- (bot1);
    \end{tikzpicture}
    } &
    \parbox[c][3.65cm][c]{4cm}{ % Adjust the size as needed
    \centering
    \begin{tikzpicture}[>=Stealth, node distance=2cm, every node/.style={circle, draw, thick}]

    \tikzset{reg/.style={draw = none,font=\sffamily}}

    %Nodes
    \node[reg] (X) at (0,1) {$X$};
    \node[reg] (Z) at (0,-1) {$Z$};
    \node[reg] (Y) at (2,0) {$Y$};
    \node[reg] (top1) at (1,1.5) {};
    \node[reg] (bot1) at (1,-1.5) {};
    \node[reg] (label1) at (0,-2) {$\tau = 1$};
    \node[reg] (label2) at (2,-2) {$\tau = 2$};
    %Edges
    \draw[->, thick] (Z) -- (X);
    \draw[->, thick] (Z) -- (Y);
    %Tiers
    \draw[-,dotted]{} (top1) -- (bot1);
    \end{tikzpicture}
    } &
    \parbox[c][3.65cm][c]{2.5cm}{ % Adjust the size as needed
    \centering
    Insert $\left\{X \rightarrow Y\right\}$
    } \\
    \hline
    \textbf{Case D.1.1.} & 
    \parbox[c][3.65cm][c]{4cm}{ % Adjust the size as needed
    \centering
    \begin{tikzpicture}[>=Stealth, node distance=2cm, every node/.style={circle, draw, thick}]

  \tikzset{reg/.style={draw = none,font=\sffamily}}

%Nodes
  \node[reg] (X) at (0,1) {$X$};
  \node[reg] (Z) at (0,-1) {$Z$};
  \node[reg] (Y) at (2,0) {$Y$};
  \node[reg] (top1) at (1,1.5) {};
  \node[reg] (bot1) at (1,-1.5) {};
  \node[reg] (label1) at (0,-2) {$\tau = 1$};
  \node[reg] (label2) at (2,-2) {$\tau = 2$};
%Edges
   \draw[->, thick] (Z) -- (Y);
%Tiers
  \draw[-,dotted]{} (top1) -- (bot1);

\end{tikzpicture}
    } &
    \parbox[c][3.65cm][c]{4cm}{ % Adjust the size as needed
    \centering
    \begin{tikzpicture}[>=Stealth, node distance=2cm, every node/.style={circle, draw, thick}]

  \tikzset{reg/.style={draw = none,font=\sffamily}}

%Nodes
  \node[reg] (X) at (0,1) {$X$};
  \node[reg] (Z) at (0,-1) {$Z$};
  \node[reg] (Y) at (2,0) {$Y$};
  \node[reg] (top1) at (1,1.5) {};
  \node[reg] (bot1) at (1,-1.5) {};
  \node[reg] (label1) at (0,-2) {$\tau = 1$};
  \node[reg] (label2) at (2,-2) {$\tau = 2$};
%Edges
   \draw[->, thick] (Z) -- (X);
   \draw[->, thick] (Z) -- (Y);
%Tiers
  \draw[-,dotted]{} (top1) -- (bot1);

\end{tikzpicture}
    } &
    \parbox[c][3.65cm][c]{2.5cm}{ % Adjust the size as needed
    \centering
    Insert $\left\{Z \rightarrow X\right\}$
    } \\
    \hline
    \textbf{Case D.1.2.1.} & 
    \parbox[c][3.65cm][c]{4cm}{ % Adjust the size as needed
    \centering
    \begin{tikzpicture}[>=Stealth, node distance=2cm, every node/.style={circle, draw, thick}]

  \tikzset{reg/.style={draw = none,font=\sffamily}}

%Nodes
  \node[reg] (X) at (0,0) {$X$};
  \node[reg] (Z) at (2,-1) {$Z$};
  \node[reg] (Y) at (2,1) {$Y$};
  \node[reg] (top1) at (1,1.5) {};
  \node[reg] (bot1) at (1,-1.5) {};
  \node[reg] (label1) at (0,-2) {$\tau = 1$};
  \node[reg] (label2) at (2,-2) {$\tau = 2$};
%Edges
   \draw[->, thick] (Z) -- (Y);
%Tiers
  \draw[-,dotted]{} (top1) -- (bot1);

\end{tikzpicture}
    } &
    \parbox[c][3.65cm][c]{4cm}{ % Adjust the size as needed
    \centering
    \begin{tikzpicture}[>=Stealth, node distance=2cm, every node/.style={circle, draw, thick}]

  \tikzset{reg/.style={draw = none,font=\sffamily}}

%Nodes
  \node[reg] (X) at (0,0) {$X$};
  \node[reg] (Z) at (2,-1) {$Z$};
  \node[reg] (Y) at (2,1) {$Y$};
  \node[reg] (top1) at (1,1.5) {};
  \node[reg] (bot1) at (1,-1.5) {};
  \node[reg] (label1) at (0,-2) {$\tau = 1$};
  \node[reg] (label2) at (2,-2) {$\tau = 2$};
%Edges
   \draw[->, thick] (X) -- (Z);
   \draw[->, thick] (Z) -- (Y);
%Tiers
  \draw[-,dotted]{} (top1) -- (bot1);

\end{tikzpicture}
    } &
    \parbox[c][3.65cm][c]{2.5cm}{ % Adjust the size as needed
    \centering
    Insert $\left\{X \rightarrow Z\right\}$
    } \\
    \hline
    \textbf{Case D.1.2.2.} & 
    \parbox[c][3.65cm][c]{4cm}{ % Adjust the size as needed
    \centering
    \begin{tikzpicture}[>=Stealth, node distance=2cm, every node/.style={circle, draw, thick}]

  \tikzset{reg/.style={draw = none,font=\sffamily}}

%Nodes
  \node[reg] (X) at (0,1) {$X$};
  \node[reg] (Z_1) at (0,-1) {$Z$};
  \node[reg] (Z_2) at (2,-1) {$W$};
  \node[reg] (Y) at (2,1) {$Y$};
  \node[reg] (top1) at (1,1.5) {};
  \node[reg] (bot1) at (1,-1.5) {};
  \node[reg] (label1) at (0,-2) {$\tau = 1$};
  \node[reg] (label2) at (2,-2) {$\tau = 2$};
%Edges
   \draw[->, thick] (Z_2) -- (Y);
   \draw[->, thick] (Z_1) -- (Z_2);
%Tiers
  \draw[-,dotted]{} (top1) -- (bot1);

\end{tikzpicture}
    } &
    \parbox[c][3.65cm][c]{4cm}{ % Adjust the size as needed
    \centering
    \begin{tikzpicture}[>=Stealth, node distance=2cm, every node/.style={circle, draw, thick}]

  \tikzset{reg/.style={draw = none,font=\sffamily}}

%Nodes
  \node[reg] (X) at (0,1) {$X$};
  \node[reg] (Z_1) at (0,-1) {$Z$};
  \node[reg] (Z_2) at (2,-1) {$W$};
  \node[reg] (Y) at (2,1) {$Y$};
  \node[reg] (top1) at (1,1.5) {};
  \node[reg] (bot1) at (1,-1.5) {};
  \node[reg] (label1) at (0,-2) {$\tau = 1$};
  \node[reg] (label2) at (2,-2) {$\tau = 2$};
%Edges
   \draw[->, thick] (Z_1) -- (X);
   \draw[->, thick] (Z_1) -- (Z_2);
   \draw[->, thick] (Z_2) -- (Y);
%Tiers
  \draw[-,dotted]{} (top1) -- (bot1);

\end{tikzpicture}
    } &
    \parbox[c][3.65cm][c]{3cm}{ % Adjust the size as needed
    \centering
    Set $Y:=W$ and $\pi := \pi_{W} = \{X,Z,W\}$. Then we are in Case D.1.1. afterwards and we insert $\left\{Z \rightarrow X\right\}$
    } \\
    \hline
    \textbf{Case D.2.} & 
    \parbox[c][3.65cm][c]{4cm}{ % Adjust the size as needed
    \centering
    \begin{tikzpicture}[>=Stealth, node distance=2cm, every node/.style={circle, draw, thick}]

  \tikzset{reg/.style={draw = none,font=\sffamily}}

%Nodes
  \node[reg] (X) at (0,1) {$X$};
  \node[reg] (Z_1) at (-1,0) {$Z_1$};
  \node[reg] (Z_2) at (0,-1) {$Z_2$};
  \node[reg] (Y) at (2,1) {$Y$};
  \node[reg] (top1) at (1,1.5) {};
  \node[reg] (bot1) at (1,-1.5) {};
  \node[reg] (label1) at (0,-2) {$\tau = 1$};
  \node[reg] (label2) at (2,-2) {$\tau = 2$};
%Edges
   \draw[->, thick] (Z_1) -- (X);
%Tiers
  \draw[-,dotted]{} (top1) -- (bot1);

\end{tikzpicture}
    } &
    \parbox[c][3.65cm][c]{4cm}{ % Adjust the size as needed
    \centering
    \begin{tikzpicture}[>=Stealth, node distance=2cm, every node/.style={circle, draw, thick}]

  \tikzset{reg/.style={draw = none,font=\sffamily}}

%Nodes
  \node[reg] (X) at (0,1) {$X$};
  \node[reg] (Z_1) at (-1,0) {$Z_1$};
  \node[reg] (Z_2) at (0,-1) {$Z_2$};
  \node[reg] (Y) at (2,1) {$Y$};
  \node[reg] (top1) at (1,1.5) {};
  \node[reg] (bot1) at (1,-1.5) {};
  \node[reg] (label1) at (0,-2) {$\tau = 1$};
  \node[reg] (label2) at (2,-2) {$\tau = 2$};
%Edges
   \draw[->, thick] (X) -- (Z_1);
   \draw[->, thick] (Z_2) -- (Z_1);
   \draw[->, thick] (Z_2) -- (Y);
%Tiers
  \draw[-,dotted]{} (top1) -- (bot1);

\end{tikzpicture}
    } &
    \parbox[c][3.65cm][c]{2.5cm}{ % Adjust the size as needed
    \centering
    Insert $\left\{Z_2 \rightarrow Y\right\}$
    } \\
    \hline
    
    \end{tabular}
    \caption{Examples of how each \textbf{Case }in the proof of Lemma \ref{Lem:TGES_for_S_C} can occur. The first column notes the case, and the second and third column respectively contains a DAG $G$ and data-generating DAG $\GS$ with tiered ordering $\tau$ that correspond to the given case. The fourth column notes which edge addition that according to local $\K$-consistency will result in a higher score, and hence a contradiction.}
    \label{tab:Lemma1_proof_examples}
\end{table}

\section{Simulation Study} \label{supmat:sec:Sim_study}

\subsection{Metrics} \label{supmat:subsec:Metrics}
\newcommand{\bidirpic}{
\begin{tikzpicture}[>=Stealth, node distance=2cm, every node/.style={circle, draw, thick}]
  \tikzset{reg/.style={draw = none,font=\sffamily}}
  \node[reg] (A) {A};
  \node[reg] (B) [right of=A] {B};
  \draw[-] (A) -- (B);
\end{tikzpicture}
}

\newcommand{\rightdirpic}{
\begin{tikzpicture}[>=Stealth, node distance=2cm, every node/.style={circle, draw, thick}]
  \tikzset{reg/.style={draw = none,font=\sffamily}}
  \node[reg] (A) {A};
  \node[reg] (B) [right of=A] {B};
  \draw[->] (A) -- (B);
\end{tikzpicture}
}
\newcommand{\leftdirpic}{
\begin{tikzpicture}[>=Stealth, node distance=2cm, every node/.style={circle, draw, thick}]
  \tikzset{reg/.style={draw = none,font=\sffamily}}
  \node[reg] (A) {A};
  \node[reg] (B) [right of=A] {B};
  \draw[->] (B) -- (A);
\end{tikzpicture}
}
\newcommand{\nodirpic}{
\begin{tikzpicture}[>=Stealth, node distance=2cm, every node/.style={circle, draw, thick}]
  \tikzset{reg/.style={draw = none,font=\sffamily}}
  \node[reg] (A) {A};
  \node[reg] (B) [right of=A] {B};
\end{tikzpicture}
}
\newcommand{\anypic}{
\begin{tikzpicture}[>=Stealth, node distance=2cm, every node/.style={circle, draw, thick}]
  \tikzset{reg/.style={draw = none,font=\sffamily}}
  \node[reg] (A) at (0,0) {A};
  \node[reg] (B) at (2,0) {B};
  \node[reg,font=\large] (star1) at (0.35,-0.055) {*};
  \node[reg,font=\large] (star2) at (1.65,-0.055) {*};
  \draw[-] (A) -- (B); 
\end{tikzpicture}
}
\subsubsection{Metric of Adjacencies}
\begin{table}[H]
\centering
\label{tab:adjconf}
\begin{tabular}{|c|c|c|}
\hline
\parbox[c][1cm][c]{4cm}{\centering \textbf{Estimated edge}} & 
\parbox[c][1cm][c]{4cm}{\centering \textbf{True edge}} & 
\parbox[c][1cm][c]{3cm}{\centering \textbf{Result}} \\
\hline
\parbox[c][1cm][c]{3cm}{\centering \anypic} & 
\parbox[c][1cm][c]{3cm}{\centering \anypic} & 
\parbox[c][1cm][c]{3cm}{\centering TP} \\
\hline
\parbox[c][1cm][c]{3cm}{\centering \nodirpic} & 
\parbox[c][1cm][c]{3cm}{\centering \anypic} & 
\parbox[c][1cm][c]{3cm}{\centering FN} \\
\hline
\parbox[c][1cm][c]{3cm}{\centering \anypic} & 
\parbox[c][1cm][c]{3cm}{\centering \nodirpic} & 
\parbox[c][1cm][c]{3cm}{\centering FP} \\
\hline
\parbox[c][1cm][c]{3cm}{\centering \nodirpic} & 
\parbox[c][1cm][c]{3cm}{\centering \nodirpic} & 
\parbox[c][1cm][c]{3cm}{\centering TN} \\
\hline
\end{tabular}
\caption{Overview of values in confusion matrix for the metric of adjacencies, given a relation between two nodes in the estimated graph and the true graph. An edge with stars $*$ symbolizes that the edge can be undirected or directed to either side.}
\end{table}

\subsubsection{Metric of all Directions}

\begin{table}[H]
\centering
\begin{tabular}{|c|c|c|}
\hline
\parbox[c][1cm][c]{4cm}{\centering \textbf{Estimated edge}} & 
\parbox[c][1cm][c]{4cm}{\centering \textbf{True edge}} & 
\parbox[c][1cm][c]{3cm}{\centering \textbf{Result}} \\
\hline
\parbox[c][1cm][c]{3cm}{\centering \bidirpic} & 
\parbox[c][1cm][c]{3cm}{\centering \bidirpic} & 
\parbox[c][1cm][c]{3cm}{\centering TN} \\
\hline
\parbox[c][1cm][c]{3cm}{\centering \bidirpic} & 
\parbox[c][1cm][c]{3cm}{\centering \rightdirpic} & 
\parbox[c][1cm][c]{3cm}{\centering FN} \\
\hline
\parbox[c][1cm][c]{3cm}{\centering \rightdirpic} & 
\parbox[c][1cm][c]{3cm}{\centering \bidirpic} & 
\parbox[c][1cm][c]{3cm}{\centering FP} \\
\hline
\parbox[c][1cm][c]{3cm}{\centering \rightdirpic} & 
\parbox[c][1cm][c]{3cm}{\centering \rightdirpic} & 
\parbox[c][1cm][c]{3cm}{\centering TP} \\
\hline
\parbox[c][1cm][c]{3cm}{\centering \rightdirpic} & 
\parbox[c][1cm][c]{3cm}{\centering \leftdirpic} & 
\parbox[c][1cm][c]{3cm}{\centering FP, FN} \\
\hline
\parbox[c][1cm][c]{3cm}{\centering \nodirpic} & 
\parbox[c][1cm][c]{3cm}{\centering \anypic} & 
\parbox[c][1cm][c]{3cm}{\centering NA  } \\
\hline
\parbox[c][1cm][c]{3cm}{\centering \anypic} & 
\parbox[c][1cm][c]{3cm}{\centering \nodirpic} & 
\parbox[c][1cm][c]{3cm}{\centering  NA } \\
\hline
\parbox[c][1cm][c]{3cm}{\centering \nodirpic} & 
\parbox[c][1cm][c]{3cm}{\centering \nodirpic} & 
\parbox[c][1cm][c]{3cm}{\centering  NA } \\
\hline
\end{tabular}
\caption{Overview of values in confusion matrix for the metric of all directions, given a relation between two nodes in the estimated graph and the true graph. An edge with stars $*$ symbolizes that the edge can be undirected or directed to either side. NA marks cases that are not scored by this metric.}
\label{tab:dirconf}
\end{table}

\subsubsection{Metric of In-Tier Directions}
For any pair of nodes $A,B$: if $\tau(A) = \tau(B)$ the entries of the confusion matrix for the metric of in-tier directions are counted equally to the metric of all directions (see Table \ref{tab:dirconf}). If instead $\tau(A) \neq \tau(B)$ then the pair has no effect. 

\subsection{Computation Time} \label{supmat:subsec:Comptime}
We ran 50 simulations for each number of nodes between 3 and 20 and computed the average time spent.
\\
We report both the total time spent by TGES and the time spent only in stage (i) of TGES. In this way, we can both assess whether there is a general change in computation time and whether there is a change in the time spent scoring graphs by using tiered background knowledge.
\begin{figure}[!htb]
    \centering
    \includegraphics[width=0.9\linewidth]{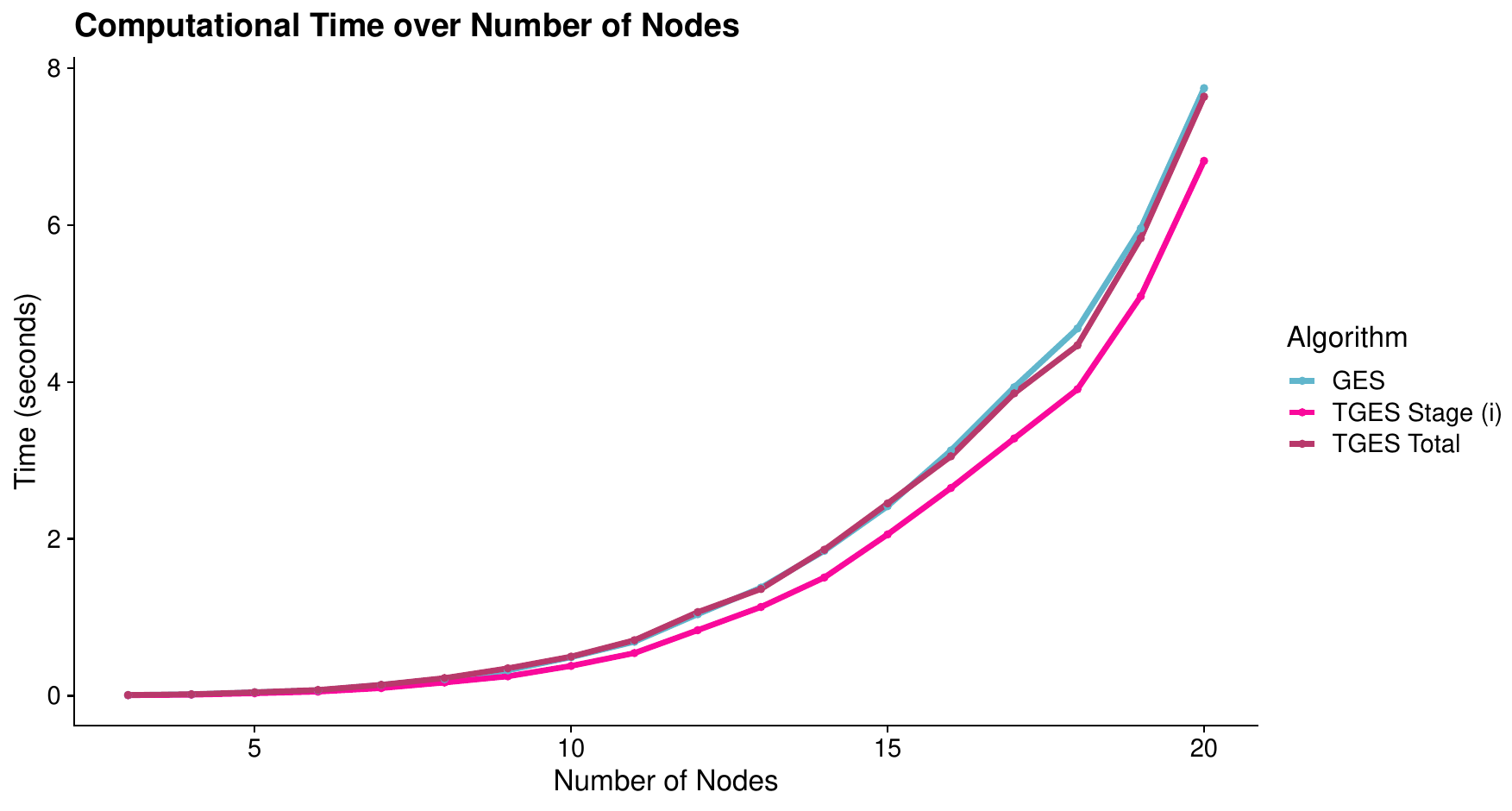}
    \caption{Computational time for a given number of nodes from 3 to 20 nodes. For each number of nodes, the computation time is the mean for 50 simulations. The times are provided for GES (as implemented in the R package \textit{pcalg} without using C++) TGES Stage (i) (time spent in Stage (i) of TGES) and TGES Total (total time spent by TGES).}
    \label{fig:Comp_time_step}
\end{figure}
\\
The total computation times spent by TGES and GES are more or less equal, as can be seen in Figure \ref{fig:Comp_time_step}. Across the number of nodes, the computation time of TGES in stage (i) is smaller than GES. This is most likely due to TGES needing to score fewer edges in the backward phase than GES as a result of using $\K$.
This shows promise of TGES being faster than GES if we are able to restrict by $\K$ and utilize Meek's rule 1 in a faster way than the current implementation.
\subsection{All methods} \label{supmat:subsec:SimStudy_AllMethods}
This section includes the plots shown in the article but with some additional methods.
\begin{figure}[H]
    \centering
    \includegraphics[width=\linewidth]{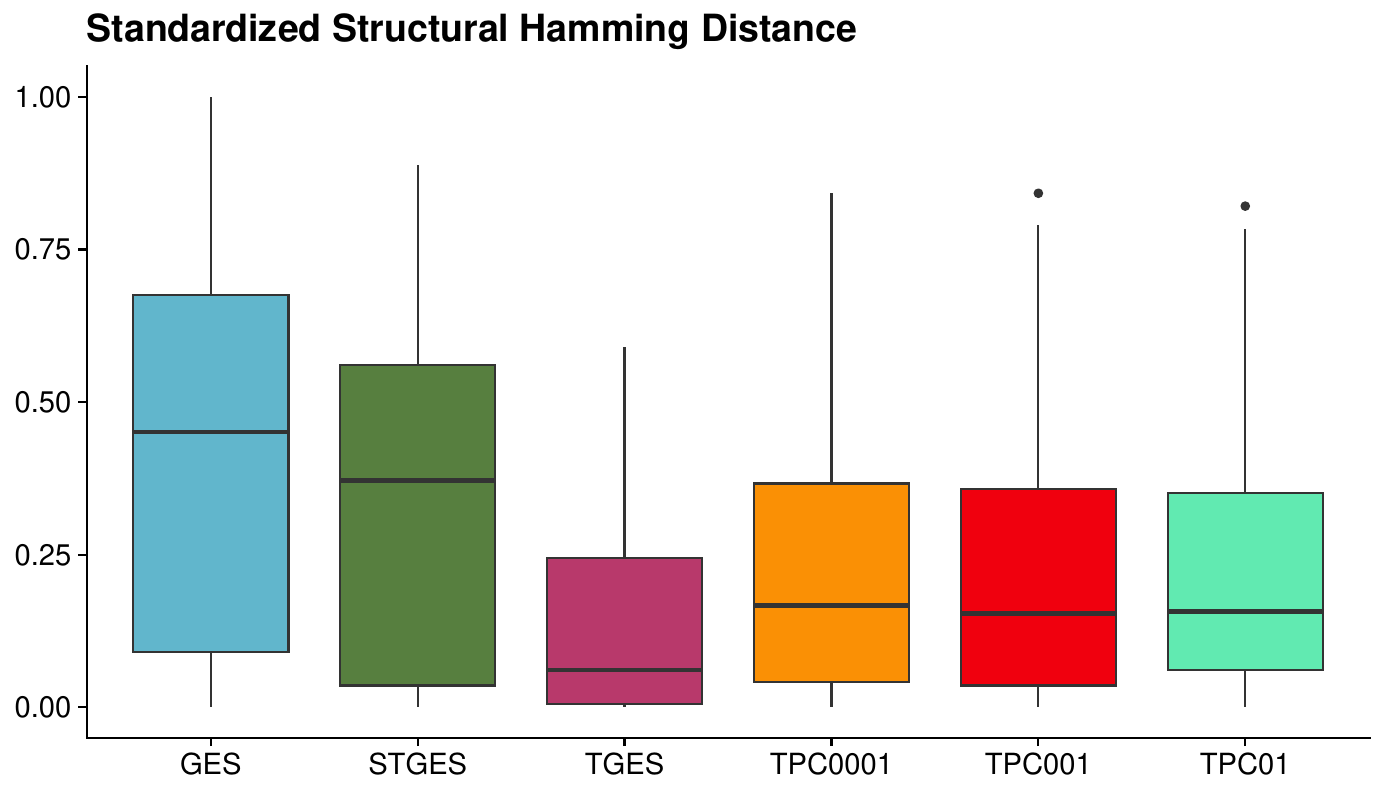}
    \caption{Standardized Structural Hamming Distance for GES, STGES, TGES, TPC with $\alpha = 0.001, 0.01, 0.1$.}
    \label{fig:App_sSHD_all}
\end{figure}
\begin{figure}[H]
    \centering
    \includegraphics[width=\linewidth]{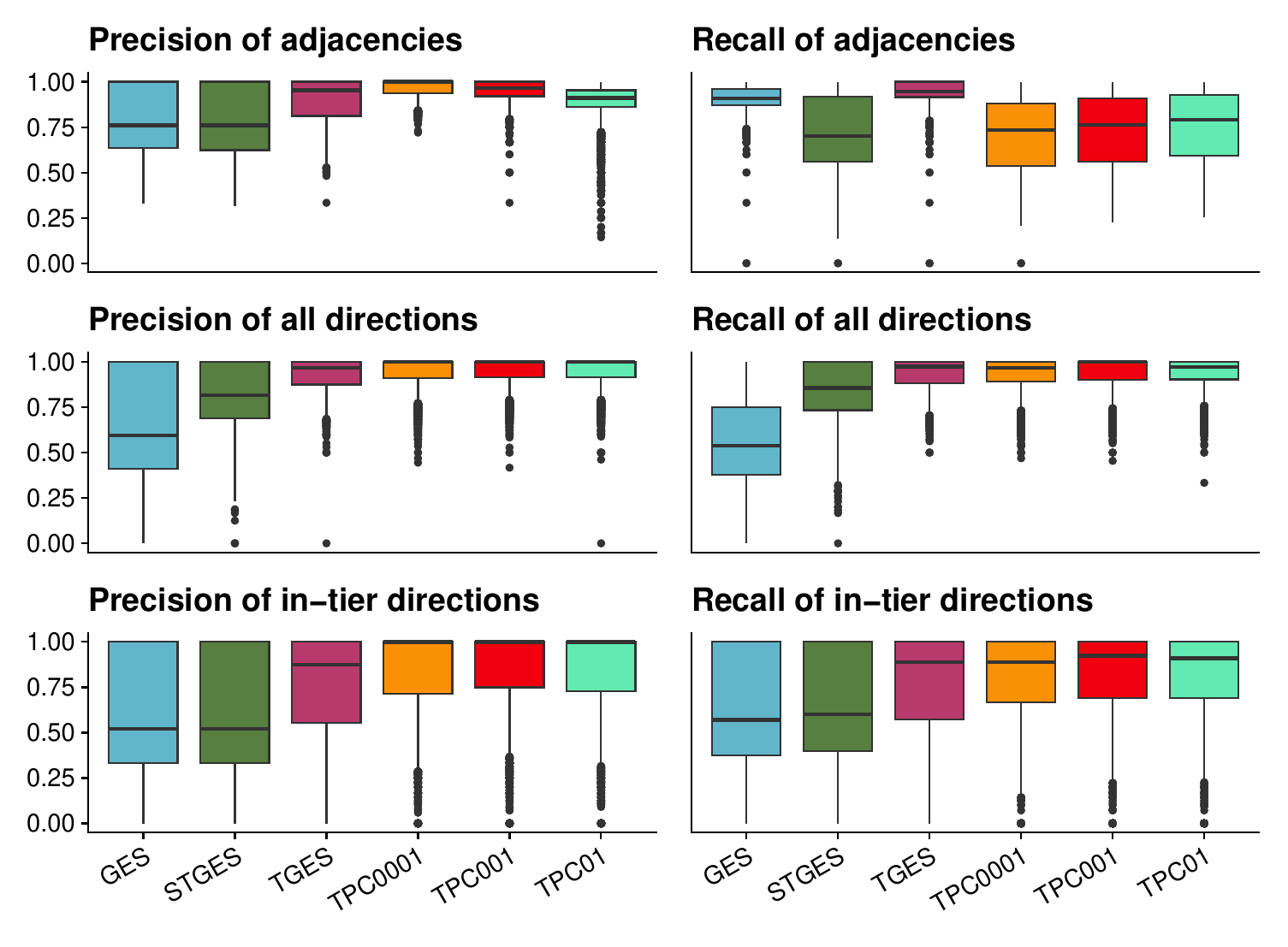}
    \caption{Metrics of adjacencies, all directions and in-tier directions for GES, STGES, TGES, TPC with $\alpha = 0.001, 0.01, 0.1$.}
    \label{fig:App_adj_dir_int_all}
\end{figure}

\section{Data Example}
\label{supmat:sec:data_example}
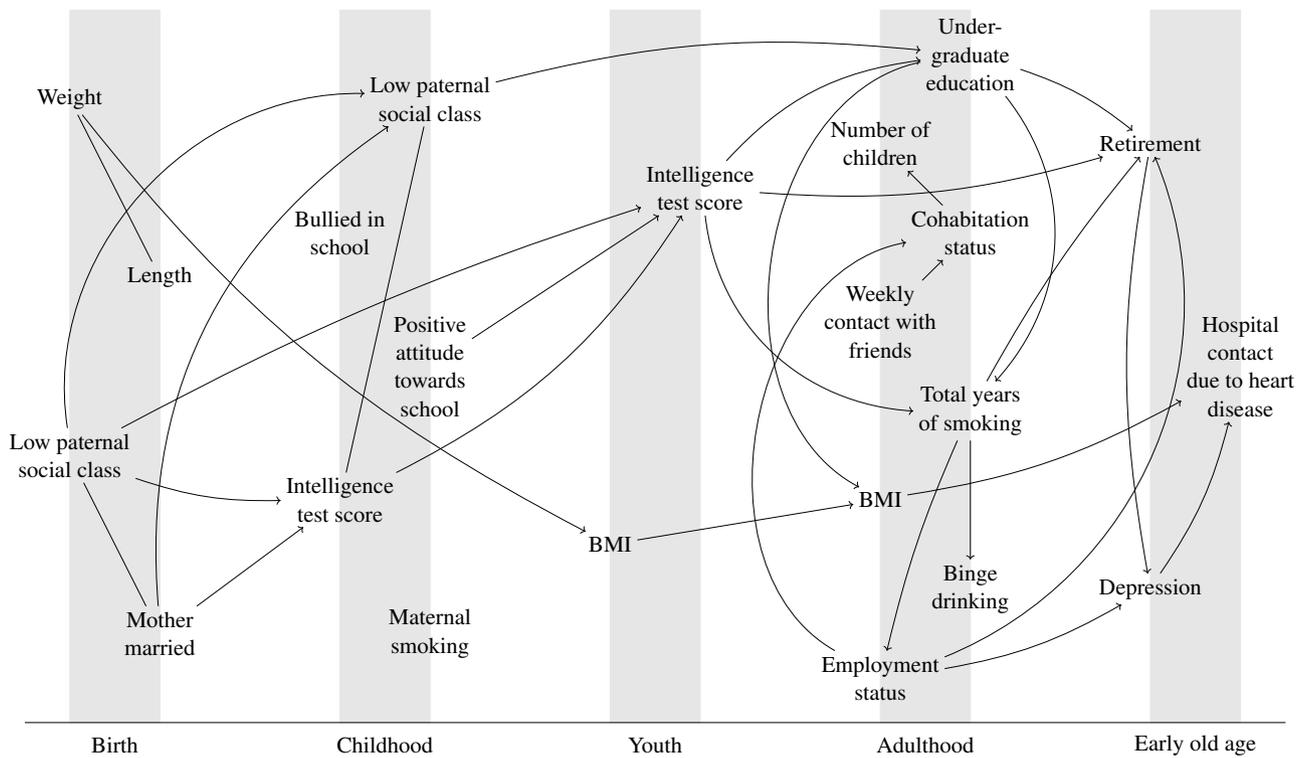
\begin{figure}[H]
    \centering
    \resizebox{\linewidth}{!}{
%TIKZ FIG MADE BY CAUSALDISCO
\begin{tikzpicture}
[every node/.style={font=\Large, align = center}, every edge/.append style={nodes={font=\itshape\scriptsize}}]
\node (1) at (2,2) {Mother \\ married};
\node at (1,-0.5) {Birth};
\node (2) at (0,6) {Low paternal \\ social class};
\node (3) at (2,10) {Length};
\node (4) at (0,14) {Weight};
\node (5) at (8,2) {Maternal \\smoking};
\node at (7,-0.5) {Childhood};
\node (6) at (6,5) {Intelligence \\ test score};
\node (7) at (8,8) {Positive \\ attitude\\ towards\\ school};
\node (8) at (6,11) {Bullied in \\school};
\node (9) at (8,14) {Low paternal \\ social class};
\node (10) at (12,4) {BMI};
\node at (13,-0.5) {Youth};
\node (11) at (14,12) {Intelligence \\test score};
\node (12) at (18,1) {Employment\\ status};
\node at (19,-0.5) {Adulthood};
\node (13) at (20,3) {Binge \\drinking};
\node (14) at (18,5) {BMI};
\node (15) at (20,7) {Total years\\ of smoking};
\node (16) at (18,9) {Weekly \\contact with\\ friends};
\node (17) at (20,11) {Cohabitation\\ status};
\node (18) at (18,13) {Number of\\ children};
\node (19) at (20,15) {Under- \\graduate\\ education};
\node (20) at (24,3) {Depression};
\node at (25,-0.5) {Early old age};
\node (21) at (26,8) {Hospital \\contact \\due to heart \\disease};
\node (22) at (24,13) {Retirement};
\draw [->] (1) edge  (6);
\draw [->] (1) edge [bend left=30] (9);
\draw [->] (2) edge [bend right=10] (6);
\draw [->] (2) edge [bend left=50] (9);
\draw [->] (2) edge [bend left=5] (11);
\draw [->] (4) edge [bend right=12] (10);
\draw [->] (6) edge [bend right=15] (11);
\draw [->] (7) edge (11);
\draw [->] (9) edge [bend left=10] (19);
\draw [->] (10) edge (14);
\draw [->] (11) edge [bend right=40] (15);
\draw [->] (11) edge [bend left=20] (19);
\draw [->] (11) edge [bend right=10] (22);
\draw [->] (12) edge [bend left=70] (17);
\draw [->] (12) edge [bend right=10] (20);
\draw [->] (12) edge [bend right=45] (22);
\draw [->] (14) edge [bend right=10] (21);
\draw [->] (15) edge [bend right=5] (12);
\draw [->] (15) edge  (13);
\draw [->] (15) edge [bend left=5] (22);
\draw [->] (16) edge  (17);
\draw [->] (17) edge  (18);
\draw [->] (19) edge [bend right=70] (14);
\draw [->] (19) edge [bend left=40] (15);
\draw [->] (19) edge [bend left=10] (22);
\draw [->] (20) edge [bend right=10] (21);
\draw [->] (22) edge [bend right=10] (20);
\draw [-] (2) edge  (1);
\draw [-] (4) edge  (3);
\draw [-] (9) edge  (6);
\draw [-] (-1,0) edge (27,0);
\begin{pgfonlayer}{background}
\filldraw [join=round,black!10]
(0,0) rectangle (2,16)
(6,0) rectangle (8,16)
(12,0) rectangle (14,16)
(18,0) rectangle (20,16)
(24,0) rectangle (26,16)
; \end{pgfonlayer}
\end{tikzpicture}
}
\caption{Tiered MPDAG fitted using TGES with score criterion TBIC with a tuned penalty parameter such that the number of edges corresponds to the number of edges in Figure \ref{fig:Expert_metro}. The data is from \citet{TheoryVsDataPetersen}.}
    \label{fig:tun_lam_TGES_metro}
\end{figure}

\begin{figure}[H]
    \centering
    \resizebox{\linewidth}{!}{
%TIKZ FIG MADE BY CAUSALDISCO
\begin{tikzpicture}
[every node/.style={font=\Large, align = center}, every edge/.append style={nodes={font=\itshape\scriptsize}}]
\node (1) at (2,2) {Mother \\ married};
\node at (1,-0.5) {Birth};
\node (2) at (0,6) {Low paternal \\ social class};
\node (3) at (2,10) {Length};
\node (4) at (0,14) {Weight};
\node (5) at (8,2) {Maternal \\smoking};
\node at (7,-0.5) {Childhood};
\node (6) at (6,5) {Intelligence \\ test score};
\node (7) at (8,8) {Positive \\ attitude\\ towards\\ school};
\node (8) at (6,11) {Bullied in \\school};
\node (9) at (8,14) {Low paternal \\ social class};
\node (10) at (12,4) {BMI};
\node at (13,-0.5) {Youth};
\node (11) at (14,12) {Intelligence \\test score};
\node (12) at (18,1) {Employment\\ status};
\node at (19,-0.5) {Adulthood};
\node (13) at (20,3) {Binge \\drinking};
\node (14) at (18,5) {BMI};
\node (15) at (20,7) {Total years\\ of smoking};
\node (16) at (18,9) {Weekly \\contact with\\ friends};
\node (17) at (20,11) {Cohabitation\\ status};
\node (18) at (18,13) {Number of\\ children};
\node (19) at (20,15) {Under- \\graduate\\ education};
\node (20) at (24,3) {Depression};
\node at (25,-0.5) {Early old age};
\node (21) at (26,8) {Hospital \\contact \\due to heart \\disease};
\node (22) at (24,13) {Retirement};
\draw [->] (1) edge [bend right=10] (5);
\draw [->] (1) edge [bend right=10] (6);
\draw [->] (1) edge [bend right=10] (9);
\draw [->] (1) edge [bend right=10] (11);
\draw [->] (2) edge [bend right=10] (1);
\draw [->] (2) edge [bend right=10] (6);
\draw [->] (2) edge [bend right=10] (9);
\draw [->] (2) edge [bend right=10] (10);
\draw [->] (2) edge [bend right=10] (11);
\draw [->] (2) edge [bend right=10] (19);
\draw [->] (3) edge [bend right=10] (1);
\draw [->] (4) edge [bend right=10] (5);
\draw [->] (4) edge [bend right=10] (10);
\draw [->] (4) edge [bend right=10] (14);
\draw [->] (5) edge [bend right=10] (6);
\draw [->] (5) edge [bend right=10] (10);
\draw [->] (5) edge [bend right=10] (15);
\draw [->] (5) edge [bend right=10] (16);
\draw [->] (5) edge [bend right=10] (21);
\draw [->] (6) edge [bend right=10] (9);
\draw [->] (6) edge [bend right=10] (11);
\draw [->] (6) edge [bend right=10] (12);
\draw [->] (6) edge [bend right=10] (15);
\draw [->] (6) edge [bend right=10] (19);
\draw [->] (7) edge [bend right=10] (6);
\draw [->] (7) edge [bend right=10] (11);
\draw [->] (7) edge [bend right=10] (18);
\draw [->] (9) edge [bend right=10] (14);
\draw [->] (9) edge [bend right=10] (19);
\draw [->] (10) edge [bend right=10] (14);
\draw [->] (10) edge [bend right=10] (15);
\draw [->] (11) edge [bend right=10] (15);
\draw [->] (11) edge [bend right=10] (19);
\draw [->] (11) edge [bend right=10] (22);
\draw [->] (12) edge [bend right=10] (17);
\draw [->] (12) edge [bend right=10] (20);
\draw [->] (12) edge [bend right=10] (22);
\draw [->] (14) edge [bend right=10] (21);
\draw [->] (14) edge [bend right=10] (22);
\draw [->] (15) edge [bend right=10] (12);
\draw [->] (15) edge [bend right=10] (13);
\draw [->] (15) edge [bend right=10] (17);
\draw [->] (15) edge [bend right=10] (22);
\draw [->] (16) edge [bend right=10] (13);
\draw [->] (16) edge [bend right=10] (17);
\draw [->] (16) edge [bend right=10] (21);
\draw [->] (17) edge [bend right=10] (13);
\draw [->] (17) edge [bend right=10] (18);
\draw [->] (18) edge [bend right=10] (22);
\draw [->] (19) edge [bend right=10] (14);
\draw [->] (19) edge [bend right=10] (15);
\draw [->] (19) edge [bend right=10] (22);
\draw [->] (20) edge [bend right=10] (21);
\draw [->] (22) edge [bend right=10] (20);
\draw [-] (4) edge [bend right=10] (3);
\draw [-] (-1,0) edge (27,0);
\begin{pgfonlayer}{background}
\filldraw [join=round,black!10]
(0,0) rectangle (2,16)
(6,0) rectangle (8,16)
(12,0) rectangle (14,16)
(18,0) rectangle (20,16)
(24,0) rectangle (26,16)
; \end{pgfonlayer}
\end{tikzpicture}

}
\caption{Tiered MPDAG fitted using TGES with score criterion TBIC on the data from \citet{TheoryVsDataPetersen}.}
    \label{fig:std_lam_TGES_metro}
\end{figure}
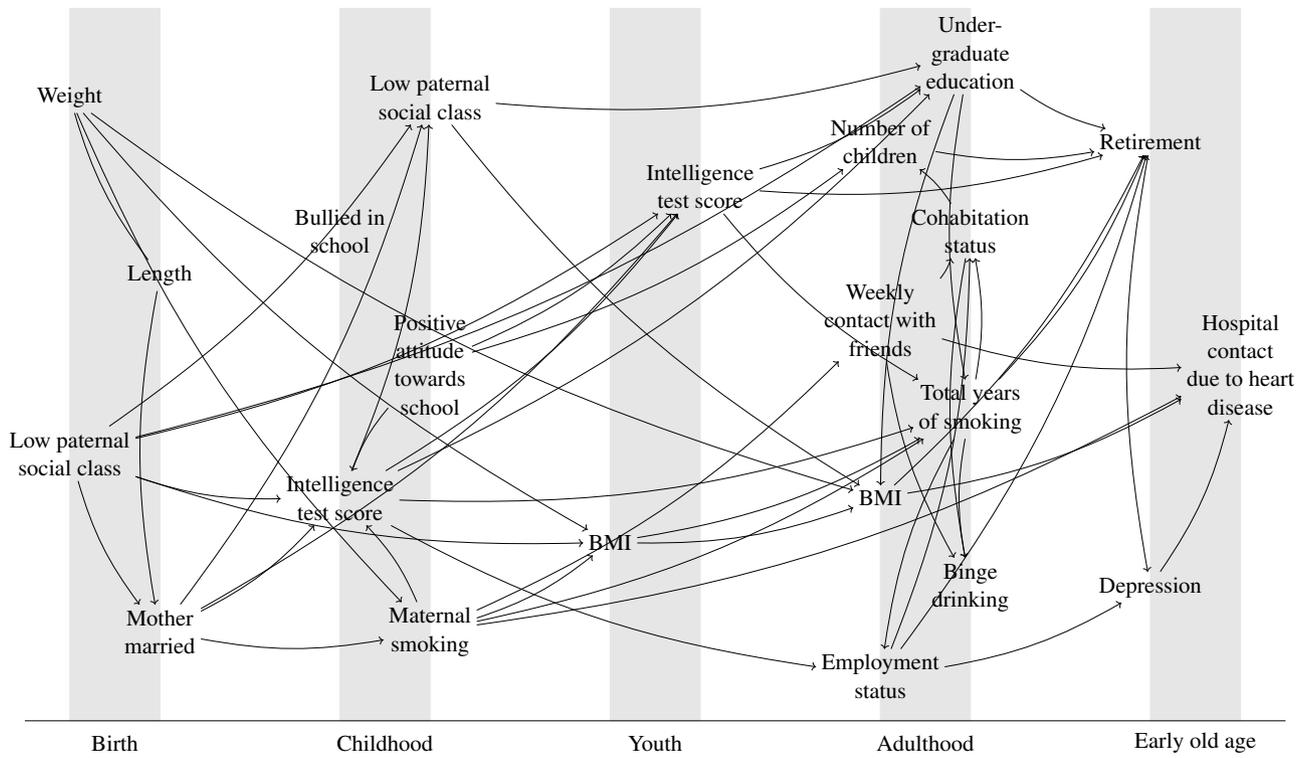

\begin{figure}[H]
    \centering
    \resizebox{\linewidth}{!}{
%TIKZ FIG MADE BY CAUSALDISCO
\begin{tikzpicture}
[every node/.style={font=\Large, align = center}, every edge/.append style={nodes={font=\itshape\scriptsize}}]
\node (1) at (2,2) {Mother \\ married};
\node at (1,-0.5) {Birth};
\node (2) at (0,6) {Low paternal \\ social class};
\node (3) at (2,10) {Length};
\node (4) at (0,14) {Weight};
\node (5) at (8,2) {Maternal \\smoking};
\node at (7,-0.5) {Childhood};
\node (6) at (6,5) {Intelligence \\ test score};
\node (7) at (8,8) {Positive \\ attitude\\ towards\\ school};
\node (8) at (6,11) {Bullied in \\school};
\node (9) at (8,14) {Low paternal \\ social class};
\node (10) at (12,4) {BMI};
\node at (13,-0.5) {Youth};
\node (11) at (14,12) {Intelligence \\test score};
\node (12) at (18,1) {Employment\\ status};
\node at (19,-0.5) {Adulthood};
\node (13) at (20,3) {Binge \\drinking};
\node (14) at (18,5) {BMI};
\node (15) at (20,7) {Total years\\ of smoking};
\node (16) at (18,9) {Weekly \\contact with\\ friends};
\node (17) at (20,11) {Cohabitation\\ status};
\node (18) at (18,13) {Number of\\ children};
\node (19) at (20,15) {Under- \\graduate\\ education};
\node (20) at (24,3) {Depression};
\node at (25,-0.5) {Early old age};
\node (21) at (26,8) {Hospital \\contact \\due to heart \\disease};
\node (22) at (24,13) {Retirement};
\draw [->] (2) edge [bend left=50]  (9);
\draw [->] (3) edge   (4);
\draw [->] (4) edge [bend left=20]  (6);
\draw [->] (5) edge [bend right=10]  (15);
\draw [->] (6) edge   (7);
\draw [->] (6) edge   (10);
\draw [->] (6) edge [bend right=15]  (11);
\draw [->] (7) edge [bend left=30]  (19);
\draw [->] (8) edge   (7);
\draw [->] (10) edge  (14);
\draw [->] (11) edge   (10);
\draw [->] (11) edge  (18);
\draw [->] (11) edge [bend left=10]  (19);
\draw [->] (12) edge   (20);
\draw [->] (12) edge [bend right=40]  (22);
\draw [->] (13) edge   (20);
\draw [->] (13) edge   (21);
\draw [->] (14) edge   (20);
\draw [->] (14) edge   (21);
\draw [->] (15) edge   (21);
\draw [->] (16) edge [bend right=35]  (20);
\draw [->] (16) edge   (21);
\draw [->] (17) edge   (20);
\draw [->] (17) edge   (21);
\draw [->] (18) edge [bend left=44]  (20);
\draw [->] (18) edge [bend left=10]  (21);
\draw [->] (19) edge [bend right=70]  (12);
\draw [->] (21) edge [bend left=10]  (20);
\draw [->] (21) edge [bend right=10]  (22);
\draw [->] (22) edge [bend right=10]  (20);
\draw [-] (-1,0) edge (27,0);
\begin{pgfonlayer}{background}
\filldraw [join=round,black!10]
(0,0) rectangle (2,16)
(6,0) rectangle (8,16)
(12,0) rectangle (14,16)
(18,0) rectangle (20,16)
(24,0) rectangle (26,16)
; \end{pgfonlayer}
\end{tikzpicture}

}
\caption{Expert consensus DAG from \citet{TheoryVsDataPetersen}. }
    \label{fig:Expert_metro}
\end{figure}
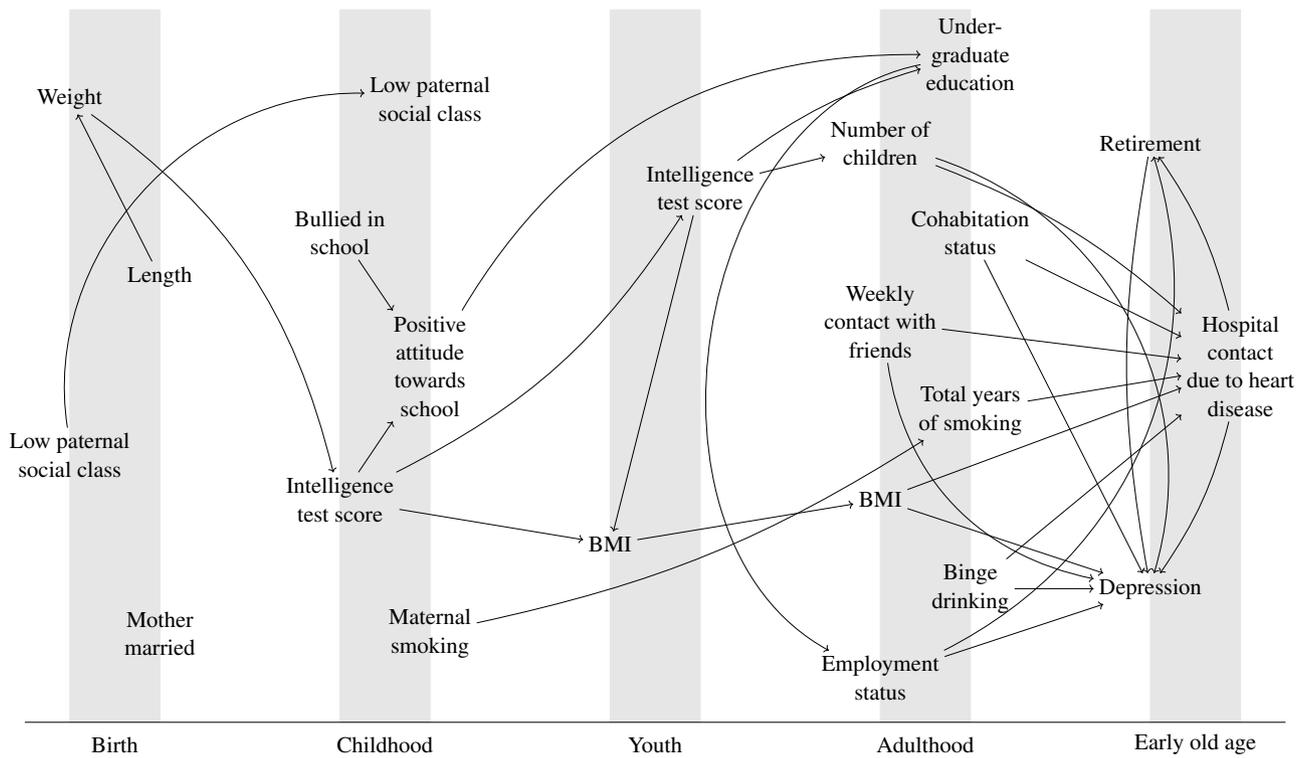

\begin{table}[H]
\centering
\begin{tabular}{lrr}
                    & \multicolumn{2}{c}{\textbf{Expert}} \\ %\cline{1-3} 
            \textbf{TPC}            & Adjacency  & Non-adjacency \\ \hline
 Adjacency       & 10         & 20            \\
 Non-adjacency & 20         & 181          
\end{tabular}
\caption{Confusion matrix from \citet{TheoryVsDataPetersen} for the TPC algorithm. TPC is tuned to fit a graph with exactly 30 edges.}
\label{tab:adj_conf_TPC}
\end{table}

\begin{table}[H]
\centering
\begin{tabular}{lrr}
                    & \multicolumn{2}{c}{\textbf{Expert}} \\ %\cline{1-3} 
            \textbf{$\text{TGES}_{\text{tuned}}$}            & True  & False \\ \hline
 Positive       & 8         & 1           \\
 Negative & 0         & 2          
\end{tabular}
\caption{Confusion matrix for the metric of all directions. Estimated using TGES with TBIC with tuned penalty parameter to match the number of edges in the expert graph.}
\label{tab:all_dir_conf_tges_tuned}
\end{table}

\begin{table}[H]
\centering
\begin{tabular}{lrr}
                    & \multicolumn{2}{c}{\textbf{Expert}} \\ %\cline{1-3} 
            \textbf{$\text{TGES}_{\text{tuned}}$}            & True  & False \\ \hline
 Positive       & 1         & 1            \\
 Negative & 0         & 2          
\end{tabular}
\caption{Confusion matrix for the metric of in-tier directions. Estimated using TGES with TBIC with tuned penalty parameter to match the number of edges in the expert graph. The estimated graph has 12 in-tier edges out of 30 edges in total.}
\label{tab:in_dir_conf_tges_tuned}
\end{table}

\begin{table}[H]
\centering
\begin{tabular}{lrr}
                    & \multicolumn{2}{c}{\textbf{Expert}} \\ %\cline{1-3} 
            \textbf{$\text{TGES}_{\text{fixed}}$}            & True  & False \\ \hline
 Positive       & 10          & 2           \\
 Negative & 0         & 3          
\end{tabular}
\caption{Confusion matrix for the metric of all directions. Estimated using TGES with TBIC.}
\label{tab:all_dir_conf_tges_fixed}
\end{table}

\begin{table}[H]
\centering
\begin{tabular}{lrr}
                    & \multicolumn{2}{c}{\textbf{Expert}} \\ %\cline{1-3} 
            \textbf{$\text{TGES}_{\text{fixed}}$}            & True  & False \\ \hline
 Positive       & 1         & 2            \\
 Negative & 0         & 3          
\end{tabular}
\caption{Confusion matrix for the metric of in-tier directions. Estimated using TGES with TBIC. The estimated graph has 18 in-tier edges out of 55 edges in total.}
\label{tab:in_dir_conf_tges_fixed}
\end{table}

\end{document}